\documentclass[twocolumn]{IEEEtran}

\usepackage{amsmath, graphics,amssymb,epsfig,subfigure}
\usepackage{algorithm}
\usepackage{algorithmic}
\usepackage{float}
\usepackage{multirow}
\usepackage{cite}

\newlength{\capwidth}
\newtheorem{Theorem}{Theorem}

\newtheorem{Lemma}{Lemma}

\newcommand{\yv}{\mbox{$\bf y $}}
\newcommand{\av}{\mbox{$\bf a $}}
\newcommand{\bv}{\mbox{$\bf b $}}

\newcommand{\Hv}{\mbox{$\bf H $}}

\newcommand{\Av}{\mbox{$\bf A $}}
\newcommand{\xv}{\mbox{$\bf x $}}

\newcommand{\Xv}{\mbox{$\bf X $}}
\newcommand{\Pv}{\mbox{$\bf P $}}

\newcommand{\nv}{\mbox{$\bf n $}}
\newcommand{\cv}{\mbox{$\bf c $}}
\newcommand{\Cv}{\mbox{$\bf C $}}

\newcommand{\ev}{\mbox{$\bf e $}}
\newcommand{\Bv}{\mbox{$\bf B $}}
\newcommand{\Qv}{\mbox{$\bf Q $}}

\newcommand{\Wv}{\mbox{$\bf W $}}

\newcommand{\Iv}{\mbox{$\bf I $}}
\newcommand{\Rv}{\mbox{$\bf R $}}

\newcommand{\Lv}{\mbox{$\bf L $}}
\newcommand{\Gv}{\mbox{$\bf G $}}

\newcommand{\uv}{\mbox{$\bf u $}}

\newcommand{\zv}{\mbox{$\bf z $}}

\newcommand{\Uv}{\mbox{$\bf U $}}
\newcommand{\Vv}{\mbox{$\bf V $}}

\newcommand{\be}{\begin{equation}}

\newcommand{\ee}{\end{equation}}
\newcommand{\bea}{\begin{eqnarray}}
\newcommand{\eea}{\end{eqnarray}}
\newcommand{\bdp}{\begin{displaymath}}
\newcommand{\edp}{\end{displaymath}}

\begin{document}
%%%%%%%%%%%%%%%%%%%%%%%%%%%%%%%%%%%%%%%%%%%%%%%%%%%%%%%%%%%%%%%%%%%%%%%%%%%%%%%%%%%%%%%%%%%%%%% Title
\title{\huge{Diversity of Linear Transceivers in MIMO AF Half-duplex Relaying Channels }}

%%%%%%%%%%%%%%%%%%%%%%%%%%%%%%%%%%%%%%%%%%%%%%%%%%%%%%%%%%%%%%%%%%%%%%%%%%%%%%%%%%%%%%%%%%%%%%
\author{\IEEEauthorblockN{Changick Song and Cong Ling} \\
\thanks{This work was supported in part by FP7 project PHYLAWS (EU FP7-ICT 317562) and partially submitted to IEEE International Symposium on Information Theory (ISIT) 2014.

The authors are with the Department of Electrical and Electronic Engineering, Imperial College, London, UK
(e-mail: \{csong and c.ling\}@imperial.ac.uk).}
}\maketitle

\begin{abstract}
Linear transceiving schemes between the relay and the destination have recently attracted much interest in MIMO
amplify-and-forward (AF) relaying systems due to low implementation complexity.
In this paper, we provide comprehensive analysis on the diversity order of the linear zero-forcing (ZF) and minimum mean squared error (MMSE) transceivers.
Firstly, we obtain a compact closed-form expression for the diversity-multiplexing tradeoff (DMT) through tight upper and lower bounds.
While our DMT analysis accurately predicts the performance of the ZF transceivers,
it is observed that the MMSE transceivers exhibit a complicated rate dependent behavior,
and thus are very unpredictable via DMT for finite rate cases.
Secondly, we highlight this interesting behavior of the MMSE transceivers and characterize the diversity order at all finite rates.
This leads to a closed-form expression for the diversity-rate tradeoff (DRT) which reveals the relationship between
the diversity, the rate, and the number of antennas at each node.
Our DRT analysis compliments our previous work on DMT,
thereby providing a complete understanding on the diversity order of linear transceiving schemes in MIMO AF relaying channels.
\end{abstract}

%%%%%%%%%%%%%%%%%%%%%%%%%%%%%%%%%%%%%%%%%%%%%%%%%%%%%%%%%%%%%%%%%%%%%%%%%%%%%%%%%%%%%%%%%%%%%%%%%%%%
%%%%%%%%%%%%%%%%%%%%%%%%%%%%%%%%%   Introduction  %%%%%%%%%%%%%%%%%%%%%%%%%%%%%%%%%%%%%%%%%%%%%%%%%%
%%%%%%%%%%%%%%%%%%%%%%%%%%%%%%%%%%%%%%%%%%%%%%%%%%%%%%%%%%%%%%%%%%%%%%%%%%%%%%%%%%%%%%%%%%%%%%%%%%%%
\section{Introduction}\label{sec:introduction}

Recently, there has been growing interest in multiple-input multiple-output (MIMO) relaying techniques,
due to combined benefits of improved link performance from MIMO channels and coverage extension from relaying techniques.
In particular, although suboptimal, linear transceivers between the relay and the destination with the amplify-and-forward (AF) strategy
have attracted much attention for its low complexity implementation \cite{Guan:08,Changick:09TWC,CXing:10a,Rong:09,HJChoi:14}.
From a system design perspective, in order to find the operating points of the system and predict its performance,
analytical research on these transceiving schemes is highly motivated \cite{BChalise:09,Louie:09,Changick:11TCOM,MAhn:14,Changick:12JSAC},
but the performance has not been fully understood yet.

The {\it ``diversity-multiplexing tradeoff''} (DMT) analysis provides a fundamental criterion to evaluate the performance of MIMO systems since
it compactly characterizes the tradeoff between the rate and the block error probability \cite{Zheng:03}.
For this reason, a large amount of research has been conducted in MIMO relaying systems based on the DMT \cite{Azarian:05,SYang:07,SYang:07a,Yuksel:07,DGunduz:10,Oleveque:10,SLoyka:10,SYang:11}.
With the minimum mean squared error (MMSE) strategy, however, the DMT is not sufficient to characterize the diversity order, because
the DMT framework, as an asymptotic notion in the high signal-to-noise ratio (SNR) and high spectral efficiency regime,
cannot distinguish between different spectral efficiencies that correspond to the same multiplexing gain which we denote by $r$.
In fact, it was shown in point-to-point (P2P) MIMO channels
that while the DMT analysis accurately predicts the diversity behavior of the MMSE receiver for the positive multiplexing gain ($r>0$),
the extrapolation of the DMT to $r=0$ is unable to predict the performance especially at low rates.
This rate-dependent behavior of MMSE receivers has first been observed by Hedayat {\it et al.} in \cite{Hedayat:07}
and comprehensively analyzed by Mehana {\it et al.} in \cite{Mehana:12}.
A similar phenomenon can be observed in MIMO AF relaying systems, but
the analysis has not been made so far.

In the first part of the paper, we introduce a new design framework for linear transceiver optimization
in MIMO AF relaying systems utilizing the error covariance decomposition (ECD).
We would like to mention that the ECD approach was first suggested in \cite{Changick:09TWC}
for designing the MMSE transceiver under the assumption that the number of data-streams
is smaller than or equal to that of relay and destination antennas.
In fact, however, any restriction on the antenna configuration is unnecessary under the MMSE strategy,
because a certain diversity gain is always achievable as the rate becomes smaller
(this phenomenon will be addressed later in the analysis part).
Therefore, it is important to provide a new result of the ECD that can be applied to any kinds of antenna configurations.
We remark that our new approach not only generalizes the previous work in \cite{Changick:09TWC}, but also
provides a ECD framework which brings the ZF and MMSE strategies together.

In the second part of the paper, we present asymptotic analysis of the aforementioned linear transceivers.
We first focus on the DMT performance of the systems.
Previously, some DMT bounds have been found in \cite{Changick:12JSAC} for the MMSE transceiver, but they are loose in general.
In this paper, we characterize the exact DMT performance of the ZF and MMSE transceivers
as a closed-form expression through deriving tight upper and lower bounds.
Note that for the sake of comparison, our analysis also covers the naive schemes
where only a constant gain factor is applied at the relay without channel state information (CSI).
The resulting DMT reveals that all linear transceiving and naive schemes are suboptimal
in terms of the achievable diversity in MIMO relaying channels \cite{DGunduz:10} \cite{Oleveque:10}.
It is also shown that while the DMT is determined by the first-hop link for the linear transceivers,
the naive schemes depend on the minimum DMT of the first-hop and second-hop MIMO links,
which implies that the naive scheme is always inferior to the transceiving scheme, especially when the number of relay antennas is large.

While our DMT analysis accurately predicts the ZF transceivers,
it is observed that when the rate is finite, the MMSE transceivers exhibit a complicated rate dependent behavior, and thus are very unpredictable via DMT.
To address this issue, we alternatively approach the outage probability of the MMSE transceiver by setting the multiplexing gain zero.
This leads to a closed-form expression for the {\it diversity-rate tradeoff} (DRT) which reveals the relationship between
the diversity, spectral efficiency, and the number of antennas at each node.
We note that under the DRT formulation, the analysis must be conducted more carefully compared to the DMT since
certain ratios and terms that were simply ignored in the DMT analysis may be relevant.
The presented bounds are tight except some discontinuity rate points, and thus can precisely predict the
diversity behavior of the MMSE transceiver.
Interestingly, we observe that as the rate becomes smaller, the MMSE transceiver approaches
the maximum likelihood (ML) performance \cite{Tang:07} \cite{Olga:07} with full-diversity order of the MIMO relay channel.
In contrast, however, the full-diversity order may not be achievable with the naive-MMSE scheme
no matter how small the rate is, which reveals the importance of
the CSI at the relay for obtaining a proper diversity gain in MMSE-based relaying systems.
Our DMT and DRT analyses are complementary to each other, thereby
allowing us to obtain a complete understanding on the diversity order of the linear transceivers in MIMO AF relay channels.
Finally, some simulations results are presented to demonstrate the accuracy of the analysis.

\textbf{Notations}: Throughout this paper, normal letters represent scalar quantities,
boldface letters indicate vectors and boldface uppercase letters designate matrices.
$\Iv_N$ is an $N \times N$ identity matrix
We use $\mathbb{C}$ and $\mathbb{S}_+^M$ to denote a set of complex numbers and $M\times M$ positive definite matrices, respectively.
$\preceq$ or $\succeq$ represent generalized inequality defined on the positive definite cone.
In addition, $E[\cdot]$, $(\cdot)^{H}$, $(\cdot)^+$, $\lceil\cdot\rceil$, and $\lfloor\cdot\rfloor$ stand for expectation,
conjugate transpose, $\max(\cdot,0)$, rounding up, and down operations, respectively.
$[\Av]_{k,k}$ and $\textrm{Tr}\left(\Av\right)$ denote the $k$-th diagonal element and trace function of a matrix $\Av$.
The $k$-th element of a vector $\av$ is denoted by $a_k$.
We denote $f(\rho)\doteq g(\rho)$, when two functions $f(\rho)$ and $g(\rho)$ are exponentially equal as
$\lim_{\rho\rightarrow\infty}\frac{\log f(\rho)}{\log \rho}=\lim_{\rho\rightarrow\infty}\frac{\log g(\rho)}{\log \rho}$.
Inequalities $\dot{\leq}$ and $\dot{\geq}$ are similarly defined.

\section{System Model}\label{sec:System Model}

%%%%%%%%%%%%%%%%%%%%%%%%%%%%%%%%%%%%%%%%%%%%%%%%%%%%%%%%%%%%%%%%%%Figure
\begin{figure}
\begin{center}
\includegraphics[width=3.4in]{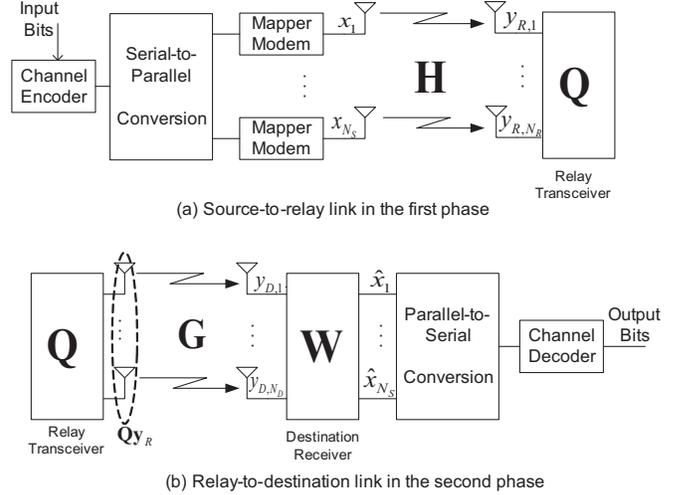}
\end{center}
\caption{Joint encoding/decoding structure for MIMO AF relaying channels with linear transceivers \label{figure:SystemModel2.eps}}
\end{figure}
%%%%%%%%%%%%%%%%%%%%%%%%%%%%%%%%%%%%%%%%%%%%%%%%%%%%%%%%%%%%%%%%%%%%%%%%

In this paper, we consider quasi-static flat fading MIMO AF relaying channels equipped with
$N_S$, $N_R$, and $N_D$ number of antennas at the source, the relay, and the destination, respectively.
A direct link between the source and the destination is ignored due to large pathloss.
We assume the half-duplex relay, which means that
each data transmission occurs in two separate phases (time or frequency).
We assume that no channel state information (CSI) is available at the source,
but the global CSI, i.e., perfect knowledge of both $\Hv$ and $\Gv$ is allowed at the destination.
The relay can be informed of either the global CSI or no CSI.

In the first phase, the output of the underlying MIMO channel between the source and the relay can be expressed as
$\yv_{R}=\Hv\xv+\nv_{R}$,
where $\xv\in\mathbb{C}^{N_S\times1}$, $\Hv\in\mathbb{C}^{N_R\times N_S}$ and $\nv_{R}\in\mathbb{C}^{N_R\times 1}$
represent the input signal vector, the channel matrix between the source and the relay, and the noise vector at the relay, respectively.
Denoting the total transmit power at the source by $P_S$, we suppose that each source antenna uses equal power
$\rho\triangleq E[| x_k|^2]=P_S/N_S$ for all $k$ because of no CSI at the source.

In the second phase, the relay signal $\yv_R$ is amplified
by the relay matrix $\Qv\in\mathbb{C}^{N_R\times N_R}$ and transmitted to the destination.
Then, the standard baseband signal at the destination is written by
\bea\label{eq:yD}
\yv_{D}=\Gv\Qv\yv_R+\nv_D=\Gv\Qv\Hv\xv+\Gv\Qv\nv_R+\nv_{D},
\eea
where $\nv_{D}$ designates the noise vector at the destination.
Note that the relay matrix $\Qv$ must
satisfy the relay power constraint $P_R$ as $E[\Vert\Qv\yv_R\Vert^2]\leq P_R$.
Finally, when a linear equalizer $\Wv\in\mathbb{C}^{N_S\times N_D}$ is employed at the destination,
the estimated signal waveform $\hat{\xv}\in\mathbb{C}^{N_S\times 1}$ is expressed as $\hat{\xv}=\Wv\yv_D$.

As the equalizer $\Wv$ decouples the received signal into $N_S$ parallel data streams,
the transmit signals at the source can be encoded either jointly or separately.
To be specific, the joint encoding indicates the case in which
a single channel encoder supports all the data streams at the source so that coding is applied jointly across antennas as illustrated in Figure \ref{figure:SystemModel2.eps}.
Hence, this coding scheme is advantageous to attain the diversity gain of MIMO channels.
In contrast, the separate encoding drives $N_S$ data streams independently using $N_S$ encoders, the outputs of which are fed to
$N_S$ independent decoders; thus, is based sorely on the spatial multiplexing.

In this work, we make a standard assumption that all entries of channel matrices $\Hv$ and $\Gv$ are independent and identically distributed (i.i.d.)
$\sim\mathcal{CN}(0,1)$ and remain constant during the transmission of a codeword.
All elements of the noise vectors $\nv_R$ and $\nv_{D}$ are also assumed to be i.i.d. $\sim\mathcal{CN}(0,1)$.
Finally, we define the following eigenvalue decompositions
$\Hv^H\Hv=\Uv_h\mathbf{\Lambda}_{h}\Uv_h^H$ and $\Gv^H\Gv=\Uv_g\mathbf{\Lambda}_{g}\Uv_g^H$,
where $\Uv_h\in\mathbb{C}^{N_S\times N_S}$ and $\Uv_g\in\mathbb{C}^{N_R\times N_R}$ are unitary matrices,
and $\mathbf{\Lambda}_{h}\in\mathbb{C}^{N_S\times N_S}$ and $\mathbf{\Lambda}_{g}\in\mathbb{C}^{N_R\times N_R}$
represent square diagonal matrices with eigenvalues $\lambda_{h,i}$ for $i=1,\ldots,N_S$ and $\lambda_{g,j}$
for $j=1,\ldots,N_R$, respectively. All eigenvalues are arranged in descending order.

\section{Linear Transceivers}\label{sec:Optimal Transceiver Design}

We would like to mention that the optimal MMSE transceiver between the relay and the destination was first developed in \cite{Guan:08}.
However, it is known that the approach in \cite{Guan:08} which is based on the singular-value decomposition (SVD) is cumbersome to
deal with due to the complicated structure of a compound
channel matrix and colored noise at the destination. In this
section, we introduce an alternative design method utilizing
the ECD property, which makes the analysis more tractable.
Our approach extends the previous result in \cite{Changick:09TWC}
and provides a ECD framework which brings the ZF and MMSE strategies together.

\subsection{MMSE Transceiver}\label{sec:MMSE Transceiver}

We start by defining the error vector $\ev\triangleq\hat{\xv}-\xv$ and its covariance matrix $\Rv_e\triangleq E[\ev\ev^H]$.
Then, the joint MMSE optimization problem for $\Qv$ and $\Wv$ is written by
\bea\label{eq:optimization problem}
\min_{\mathbf{Q,W}}\text{Tr}\left(\Rv_e\right)~~s.t.~~\text{Tr}\left(\Qv(\rho\Hv\Hv^H+\Iv_{N_R})\Qv^H\right)\leq P_R.
\eea
The problem is unconstrained and convex with respect to $\Wv$, and thus
the solution for $\Wv$ is easily obtained as the Wiener filter, i.e.,
\bea\label{eq:MMSE receiver}
\hat{\Wv}_{\text{WF}}=\rho\Hv^H\Qv^H\Gv^H(\rho\Gv\Qv\Hv\Hv^H\Qv^H\Gv^H+\Rv_n)^{-1},
\eea
where $\Rv_n\triangleq\Gv\Qv\Qv^H\Gv^H+\Iv_{N_D}$ designates the covariance matrix of the effective noise $\nv\triangleq\Gv\Qv\nv_R+\nv_D$.
Therefore, a principal issue of the problem (\ref{eq:optimization problem}) is to find $\Qv$.
The following lemma shows that with the MMSE strategy, the optimal relay matrix takes a particular structure.
\begin{Lemma}\label{Lemma:Lemma1}
{\it
The optimal relay matrix $\Qv$ of the problem (\ref{eq:optimization problem})
is generally expressed as a product of two matrices
\bea\label{eq:lemma1}
\hat{\Qv}=\Bv\Lv,
\eea
where $\Bv\in\mathbb{C}^{N_R\times N_S}$ is an unknown matrix as of yet, but
$\Lv\in\mathbb{C}^{N_S\times N_R}$ is a matrix which is given by $\Lv=\Pv\Hv^H$ with $\Pv\in\mathbb{C}^{N_S\times N_S}$ being an arbitrary square invertible matrix.}
\end{Lemma}
\begin{proof}
See Appendix \ref{Appendix:Proof of Lemma1}.
\end{proof}

Let us set $\Pv=(\Hv^H\Hv+\rho^{-1}\Iv_{N_S})^{-1}$ so that
$\Lv$ forms the MMSE receiver for the first-hop channel $\Hv$ with the input signal $\xv$.
Now, we define $\yv\triangleq\Lv\yv_R\in\mathbb{C}^{N_S\times1}$ as the output signal of the relay receiver $\Lv$,
and its covariance matrix $\Rv_y\triangleq E[\yv\yv^H]\in\mathbb{C}^{N_S\times N_S}$ as
\bea\label{eq:Ry}
\Rv_y=\Lv(\rho\Hv\Hv^H+\Iv_{N_R})\Lv^H.
\eea
Then, the estimated signal vector $\hat{\xv}$ and the relay power constraint in (\ref{eq:optimization problem}) are respectively rephrased as
\bea\label{eq:relay power constraint}
\hat{\xv}=\Wv(\Gv\Bv\yv+\nv_D)~~\text{and}~~\text{Tr}(\Bv\Rv_y\Bv^H)\leq P_R.
\eea
Since the rank of $\Rv_y$ equals $M\triangleq\min(N_S,N_R)$, $\Rv_y$ becomes clearly non-invertible when $N_S>N_R$.
This fact makes the problem more challenging, but has not been fully addressed in conventional literature.
In the following, we revisit the previous works in \cite{Guan:08} and \cite{Changick:09TWC}, and
provide a more generalized and insightful design strategy without restriction on the number of antennas at the source.

In fact, when the relay receiver $\Lv$ forms the MMSE receiver for the first hop channel,
one can show that the error covariance matrix $\Rv_e$ in (\ref{eq:optimization problem}) is expressed as a sum of
two individual covariance matrices, each of which represents the first hop and the second hop MIMO channels, respectively.
This result has been proved in \cite{Changick:09TWC}, but the proof was limited to the cases of $N_S\leq\min(N_R,N_D)$.
For the sake of completeness, we give a new result of error decomposition that can be applied to any kind of antenna configurations.

\begin{Lemma}\label{Lemma:Lemma2}
{\it Define the eigenvalue decomposition
$\Rv_y=\Uv_h\mathbf{\Lambda}_{y}\Uv_h^H$ where $\mathbf{\Lambda}_{y}\in\mathbb{C}^{N_S\times N_S}$ represents a square diagonal
matrix with eigenvalues $\lambda_{y,k}$ for $k=1,\ldots,N_S$ arranged in descending order.
Then, without loss of MMSE optimality, we have
\bea\label{eq:lemma2}
\Rv_e&\!\!=\!\!&(\Hv^H\Hv+\rho^{-1}\Iv_{N_S})^{-1}\nonumber\\
&&+~~\widetilde{\Uv}_h\big(\widetilde{\Uv}_h^H\Bv^H\Gv^H\Gv\Bv\widetilde{\Uv}_h+\widetilde{\mathbf{\Lambda}}_y^{-1}\big)^{-1}\widetilde{\Uv}_h^H,
\eea
where $\widetilde{\Uv}_h\in\mathbb{C}^{N_S\times M}$ is a matrix constructed by the first $M$ columns of $\Uv_h$ and
$\widetilde{\mathbf{\Lambda}}_y$ indicates
the $M\times M$ upper-left submatrix of $\mathbf{\Lambda}_{y}$.}
\end{Lemma}
\begin{IEEEproof}
As the relay receiver $\Lv$ follows the receive Wiener filter structure,
its output signal $\yv$ must satisfy the orthogonality principle \cite{Paulraj:03}, i.e., $E\left[\big(\yv-\xv\big)\yv^H\right]=\mathbf{0}$.
Meanwhile, using $\yv$, the MSE can be expressed as
$E\left[\|\ev\|^2\right]=E\left[\left\|\hat{\xv}-\yv+\yv-\xv\right\|^2\right]$.
Then, due to the orthogonality principle above, it is true that
the signal $\yv-\xv$ becomes orthogonal to $\hat{\xv}$ as well as $\yv$, since $\hat{\xv}=\Wv\yv_D=\Wv(\Gv\Bv\yv+\nv_D)$
is also a function of $\yv$ and independent noise $\nv_D$.
Therefore, we have
\bea\label{eq:two MSE}
E\left[\|\ev\|^2\right]=\text{MSE}_{H}+\text{MSE}_{G},
\eea
where $\text{MSE}_{H}\triangleq E\left[\left\|\yv-\xv\right\|^2\right]$ and $\text{MSE}_{G}\triangleq E\left[\left\|\Wv\yv_D-\yv\right\|^2\right]$.

In what follows, we will show that $\text{MSE}_H$ and $\text{MSE}_G$ in (\ref{eq:two MSE}) can be expressed
as the first and second term in (\ref{eq:lemma2}), respectively.
Let us first take a look at $\text{MSE}_G$. Then, it follows
\bea
\text{MSE}_\text{G}=E[\text{Tr}\left((\Wv\yv_D-\yv)(\Wv\yv_D-\yv)^H\right)]~~~~~~~~~~~~~~~~~~~~~\nonumber\\
=\text{Tr}\big(\Rv_y\!\!-\!\!\Rv_y\Bv^H\Gv^H(\Gv\Bv\Rv_y\Bv^H\Gv^H\!\!\!+\!\Iv_{N_D})^{-1}\Gv^H\Bv^H\Rv_y\big).~~\nonumber
\eea
Let us now expand the matrix $\Bv$ to a more general form as $\Bv=\breve{\Bv}\Uv_{h}^H$ where $\breve{\Bv}=[\Bv_1~\Bv_2]$
with $\Bv_1\in\mathbb{C}^{N_R\times M}$ and $\Bv_2\in\mathbb{C}^{N_R\times (N_S-M)}$.
Since $\Rv_y$ is a rank-$M$ matrix,
setting $\Bv_2=\mathbf{0}$ has no impact on both the MSE and the relay power consumption in (\ref{eq:relay power constraint}).
Therefore, without loss of generality, $\text{MSE}_\text{G}$ in (\ref{eq:two MSE}) is further rephrased as
\bea
\text{MSE}_\text{G}&=&\text{Tr}\big(\widetilde{\Uv}_h(\widetilde{\mathbf{\Lambda}}_y-\widetilde{\mathbf{\Lambda}}_y^H\Bv_1^H\Gv^H\nonumber\\
&&\times(\Gv\Bv_1\widetilde{\mathbf{\Lambda}}_y
\Bv_1^H\Gv^H+\Iv_{N_D})^{-1}\Gv^H\Bv_1^H\widetilde{\mathbf{\Lambda}}_y)\widetilde{\Uv}_h^H\big)\nonumber\\
&=&\widetilde{\Uv}_h\big(\Bv_1^H\Gv^H\Gv\Bv_1\widetilde{\Uv}_{h}+\widetilde{\mathbf{\Lambda}}_y^{-1}\big)^{-1}\widetilde{\Uv}_{h}^H~~~~~\nonumber\\
&=&\widetilde{\Uv}_h\big(\widetilde{\Uv}_{y}^H\Bv^H\Gv^H\Gv\Bv\widetilde{\Uv}_{h}+\widetilde{\mathbf{\Lambda}}_y^{-1}\big)^{-1}\widetilde{\Uv}_{h}^H,~\nonumber
\eea
where the last equality follows from $\Bv_1=\Bv\widetilde{\Uv}_h$.
Meanwhile, $\text{MSE}_H$ is equivalent to one in the conventional P2P MMSE systems;
thus, the proof simply follows from the previous results in \cite{Joham:05} and the proof is completed.
\end{IEEEproof}

The result of Lemma \ref{Lemma:Lemma2} reveals that we need to optimize only the second MSE term with respect to $\Bv$
because the first term of $\Rv_e$ consists of known parameters.
The standard theory of MMSE filter designs \cite{Guan:08} \cite{Palomar:03} shows that
in this case, the optimal $\Bv$ can be written in general by
$\hat{\Bv}=\widetilde{\Uv}_g\mathbf{\Phi}\widetilde{\Uv}_h^H$ where $\widetilde{\Uv}_g\in\mathbb{C}^{N_R\times M}$ denotes a matrix constructed by the first $M$ columns
of $\Uv_g$ and $\mathbf{\Phi}\in\mathbb{C}^{M\times M}$ is an arbitrary matrix.
Finally, substituting $\hat{\Bv}$ into (\ref{eq:lemma2}), the modified problem determines the optimal $\mathbf{\Phi}$:
\bea\label{eq:modified problem}
\hat{\mathbf{\Phi}}&=&\arg\min_{\mathbf{\Phi}}~\big(\mathbf{\Phi}\widetilde{\mathbf{\Lambda}}_g\mathbf{\Phi}^H+\widetilde{\mathbf{\Lambda}}_y^{-1}\big)^{-1}\nonumber\\
&&\textit{s.t.}~\text{Tr}(\mathbf{\Phi}\widetilde{\mathbf{\Lambda}}_y\mathbf{\Phi}^H)\leq P_R
\eea
with $\widetilde{\mathbf{\Lambda}}_g$ representing
the $M\times M$ upper-left submatrix of $\mathbf{\Lambda}_g$.
It is known that for $\Av$ and $\Bv\in\mathbb{S}_+^{M}$, we have
$\text{Tr}(\Av^{-1})\geq \sum_{i=1}^{M}\left([\Av]_{k,k}\right)^{-1}$ and
$\text{Tr}(\Av\Bv)\geq\sum_{i=1}^M\lambda_i(\Av)\lambda_{M-i+1}(\Bv)$ \cite{Komaroff:90}.
From these results, it is immediate that the minimum MSE is achieved when $\mathbf{\Phi}$ is a diagonal matrix and
the resulting problem simply becomes convex; thus, can be easily solved
by Karush-Kuhn-Tucker conditions \cite{Boyd:04}.
In combination with the relay receiver $\Lv$ in (\ref{eq:lemma1}), we finally have
\bea\label{eq:optimal relay matrix}
\hat{\Qv}=\hat{\Bv}\Lv=\widetilde{\Uv}_g\hat{\mathbf{\Phi}}\widetilde{\Uv}_h^H\Lv,
\eea
where the $k$-th diagonal element of $\hat{\mathbf{\Phi}}$ denoted by $\hat{\phi}_k$ is given by
$|\hat{\phi}_k|^2=(\lambda_{y,k}\lambda_{g,k})^{-1}\left(\nu(\lambda_{y,k}\lambda_{g,k})^{1/2}-1\right)^+$ for $k=1,2,\ldots,M$
with $\nu$ being chosen to satisfy the relay power constraint in (\ref{eq:relay power constraint}).
Note that if $\lambda_{g,k}=0$, we have $|\hat{\phi_k}|^2=0$.

\subsection{ZF Transceiver}\label{sec:ZF Transceiver}

As far as the CSI is allowed at the relay, it is also possible to improve the performance of ZF systems through the transceiver optimization process.
In fact, the optimal ZF transceiver may be similarly obtained using the SVD method in \cite{Guan:08}.
To the best of our knowledge, however, the solution for the ZF transceiver has not been presented explicitly so far.
Besides, the SVD approach may lead to an intractable solution.
In this section, we briefly show that the ZF transceiver can be obtained in our ECD framework
and present an explicit solution for the subsequent analysis.

The ZF problem arises from the constraint that $\hat{\xv}$ is an interference-free estimation of $\xv$. Thus, the
optimization problem can be formulated as \cite{Joham:05}
\bea\label{eq:ZF problem}
\min_{\mathbf{Q,W}}\text{Tr}\left(\Rv_e\right)~~s.t.~~\text{Tr}\left(\Qv(\rho\Hv\Hv^H+\Iv_{N_R})\Qv^H\right)\leq P_R\\
\label{eq:ZF constraint}
\Wv\Rv_n^{-1/2}\Gv\Qv\Hv=\Iv.~~~~~~~~~~~~~~~
\eea
We notice that the ZF problem is only defined when $N_S\leq\min(N_R,N_D)$ due to the ZF constraint (\ref{eq:ZF constraint}).
Then, the solution for the destination receiver is simply given by
\bea
\hat{\Wv}_{\text{ZF}}=(\Hv^H\Qv^H\Gv^H\Rv_n^{-1}\Gv\Qv\Hv)^{-1}\Hv^H\Qv^H\Gv^H\Rv_n^{-\frac{1}{2}}.\nonumber
\eea
Once $\hat{\Wv}_{\text{ZF}}$ is given, the constraint (\ref{eq:ZF constraint}) can be removed, which means that the remaining problem
for $\Qv$ amounts to the standard MMSE problem (\ref{eq:optimization problem}).
It is thus clear from Lemma \ref{Lemma:Lemma1} that by setting $\Pv=(\Hv^H\Hv)^{-1}$, the optimal relay matrix can be expressed as $\hat{\Qv}=\Bv\Lv_{z}$ where
$\Lv_{z}=(\Hv^H\Hv)^{-1}\Hv^H$ represents the ZF receiver for the first-hop channel $\Hv$, while $\Bv$ is an unknown matrix yet.

Now, let us apply the results of $\hat{\Wv}_{\text{ZF}}$ and $\hat{\Qv}$ to the problem (\ref{eq:ZF problem}). Then, the error covariance matrix is
\bea\label{eq:ZF error covariance}
\Rv_e&\!\!\!=\!\!\!&\Big(\Hv^H\hat{\Qv}^H\Gv^H(\Gv\hat{\Qv}\hat{\Qv}^H\Gv^H+\Iv_{N_D})^{-1}\Gv\hat{\Qv}\Hv\Big)^{-1}\nonumber\\
&\!\!\!=\!\!\!&\Big(\Bv^H\Gv^H\big(\Gv\Bv(\Hv^H\Hv)^{-1}\Bv^H\Gv^H+\Iv_{N_D}\big)^{-1}\Gv\Bv\Big)^{-1}\nonumber\\
&\!\!\!=\!\!\!&\Big(\Bv^H\Gv^H\big(\Iv_{N_D}-\Gv\Bv(\Hv^H\Hv+\Bv^H\Gv^H\Gv\Bv)^{-1}\nonumber\\
&&\times~\Bv^H\Gv^H\big)^{-1}\Gv\Bv\Big)^{-1}\nonumber\\
&\!\!\!=\!\!\!&\Big(\Bv^H\Gv^H\Gv\Bv-\Bv^H\Gv^H\Gv\Bv\nonumber\\
&&\times~\big(\Hv^H\Hv+\Bv^H\Gv^H\Gv\Bv\big)^{-1}\Bv^H\Gv^H\Gv\Bv\Big)^{-1}\nonumber\\
&\!\!\!=\!\!\!&\Big(\Hv^H\Hv\Big)^{-1}+\Big(\Bv^H\Gv^H\Gv\Bv\Big)^{-1}.
\eea
Similarly, the relay power constraint is rewritten as $\text{Tr}(\Bv\Rv_z\Bv^H)\leq P_R$ where
\bea\label{eq:Rz}
\Rv_z&\triangleq&\Lv_{z}(\rho\Hv\Hv^H+\Iv_{N_R})\Lv_{z}^H
\eea
indicates the covariance matrix of the relay signal $\zv=\Lv_z\yv_R$.
The results in (\ref{eq:ZF error covariance}) and (\ref{eq:Rz}) imply that
the error covariance decomposition method holds for the ZF systems as well.

From the equation (\ref{eq:ZF error covariance}), we obtain the modified problem to find $\Bv$:
\bea\label{eq:ZF B}
~~\min_{\mathbf{B}}\text{Tr}\Big(\big(\Bv^H\Gv^H\Gv\Bv\big)^{-1}\Big)~~s.t.~~\text{Tr}(\Bv\Rv_z\Bv^H)\leq P_R.\nonumber
\eea
Similar to (\ref{eq:modified problem}), we set $\hat{\Bv}=\widetilde{\Uv}_g\mathbf{\Phi}\Uv_h^H$
with $\mathbf{\Phi}_z\in\mathbb{C}^{N_S\times N_S}$ being a square diagonal matrix.
Then, the remaining steps simply follow the previous work in Section \ref{sec:MMSE Transceiver}.
Finally, we obtain the optimal relay matrix as
$\hat{\Qv}=\widetilde{\Uv}_g\hat{\mathbf{\Phi}}\Uv_h\Lv_z$
where the $k$-th diagonal element of $\hat{\mathbf{\Phi}}$ denoted by $\hat{\phi}_{k}$ is given by
%|\hat{\phi}_k|^2=\frac{P_R}{\sum_{i=1}^{N_S}\frac{\lambda_{z,i}}{\lambda_{g,i}}}\sqrt{\frac{1}{\lambda_{z,k}\lambda_{g,k}}}\nonumber
$|\hat{\phi}_{k}|^2=\sqrt{\frac{\mu}{\lambda_{z,k}\lambda_{g,k}}}$
for $k=1,2,\ldots,N_S$, $\lambda_{z,1}>\cdots>\lambda_{z,N_S}$ designate the eigenvalues of $\Rv_z$, and
$\mu$ is chosen to satisfy the relay power constraint.
Note that if $\lambda_{g,k}=0$, we have $|\hat{\phi_{k}}|^2=0$.

\subsection{Naive Schemes}
Meanwhile, when no CSI is available at the relay, a sensible transmission strategy is isotropic \cite{DGunduz:10} \cite{SYang:07}, i.e.,
$\hat{\Qv}=\delta\Iv_{N_R}$, which is called {\it``Naive-MMSE''} or {\it``Naive-ZF''} depending on the equalizer used at the destination.
The relay may use a scalar gain $\delta$ such that $\delta\leq\sqrt{\frac{P_R}{E[\yv_R\yv_R^H]}}$ to
remain within the power constraint. However, this variable gain requires estimation of the source-to-relay channel.
Alternatively, we can exploit a fixed gain relay which amplifies the received signal with a constant factor $c$, i.e., $\delta=c$ \cite{SLoyka:11}.
As will be shown later, both cases exhibit the same diversity behavior.

\section{Diversity-Multiplexing Tradeoff Analysis}\label{sec:Diversity Analysis}

The DMT analysis provides a compact characterization of the tradeoff between the data rate and block-error probability over the
MIMO quasi-static fading channels.
For this reason, the DMT has been widely exploited as a convenient tool for comparing various relaying systems with different protocols.
In this section, we aim to examine the DMT performance of the linear ZF and MMSE transceivers.
Throughout the analysis, we say that a system achieves multiplexing gain $r$ and corresponding diversity gain $d(r)$ if
\bea
\lim_{\rho\rightarrow\infty}\frac{R(\rho)}{\log\rho}\doteq r, ~~\text{and}~~\lim_{\rho\rightarrow\infty}\frac{P_{\text{out}}(\rho)}{\log\rho}\doteq d(r),\nonumber
\eea
where $R(\rho)$ denotes a certain target data rate that varies depending on the input SNR $\rho$ and
$P_{\text{out}}$ indicates the outage probability.
Note that if the rate $R(\rho)$ is a fixed constant regardless of the SNR, the multiplexing gain converges to zero.
We consider infinite length codewords so that
the error event is dominated by the outage event of mutual information (MI)\footnote{The practical finite-length code design
whose diversity order approaches the outage exponent will be discussed in our future works}.
In addition, we assume that $P_R=P_T=\rho N_t$ for simplicity, but the result can be easily extended to more general cases.
We use $N_S\times N_R \times N_D$ to denote a relaying system with $N_S$-source, $N_R$-relay, and $N_D$-destination antennas.

The optimal DMT is the best possible error probability exponent $d^*(r)$ achievable over a channel by any space-time codes at multiplexing gain $r$.
Before we proceed our analysis, we first need to establish the optimal DMT in MIMO AF half-duplex relaying channels.
Assuming the global CSI at the relay and the optimal ML receiver at the destination, the maximum end-to-end MI is expressed as \cite{Tang:07}
\bea\label{eq:max MI}
\mathcal{I}^*=\max_{\mathbf{Q}\in\mathbb{C}^{N_R\times N_R}}\frac{1}{2}\log\left|\rho\Hv^H\Qv^H\Gv^H\Rv_n^{-1}\Gv\Qv\Hv+\Iv_{N_S}\right|,
\eea
with the relay matrix $\Qv$ being subject to the power constraint in (\ref{eq:optimization problem}).
The pre-log factor $1/2$ is attributed to the half-duplex relay.
With the MI given above, we are ready to show that
\bea\label{eq:optimal DMT}
d^*(r)=(N_R-2r)(\min(N_S,N_D)-2r),
\eea
for $0< r< \frac{N}{2}$ with $N\triangleq\min(N_S,N_R,N_D)$.
The converse proof is immediate from the cut-set bound \cite{DGunduz:10}, since
the end-to-end DMT is bounded by the DMT of the source-to-relay cut and the relay-to-destination cut, each of which is a P2P MIMO channel.
As the transmission occurs over two time (frequency) phases under the half-duplex constraint, we obtain $d^*(r)$ in (\ref{eq:optimal DMT}) through simple scaling.
The achievability follows by showing that the MI in (\ref{eq:max MI}) achieves the cut-set bound.
Details are given in Appendix \ref{Appendix:Proof of Theorem1}.

%\subsection{DMT of Linear Transceiver}
While $d^*(r)$ is achievable by the ML decoding at the destination, the following result characterizes the DMT of
MIMO AF relaying channels under the linear transceivers.
With the linear ZF or MMSE equalizer at the destination,
the outage performance is characterized by the following two probabilities.
As for the joint encoding scheme, the outage probability of interest is given by
\bea\label{eq:P_JE}
P^{\text{JE}}_{\text{out}}\triangleq P\left(\frac{1}{2}\sum_{i=1}^{N_S}\log\left(1+\tau_i\right)<R(\rho)\right),
\eea
where $\tau_i$ indicates the output SNR of the $i$-th data stream.
Meanwhile, for the separate encoding scheme, a reasonable strategy without CSI at the source is to allocate the same rate $R/N_S$ to each stream.
Then, the relevant outage probability is given by
\bea\label{eq:P_SE}
P^{\text{SE}}_{\text{out}}\triangleq P\left(\bigcup_{i=1}^{N_S}\left\{\frac{1}{2}\log\left(1+\tau_i\right)<\frac{R(\rho)}{N_S}\right\}\right).
\eea
Using these outage definitions, we characterize the DMT performance of the MMSE transceiver as a closed-form expression in the following theorem.

\begin{Theorem}\label{Theorem:Theorem1}
{\it The DMT of the $N_S\times N_R\times N_D$ MIMO AF half-duplex relaying channels under the MMSE transceiver is given by
\bea
d(r)=\left\{\begin{array}{cc}\!\!\!\!(N_R-N_S+1)\left(1-\frac{2r}{N_S}\right)^+ & \!\!\!\!\!\text{if}~ N_S\leq \min(N_R,N_D) \\
0 & \text{otherwise}\end{array}\right.\nonumber
\eea
\vspace{-15pt}
\bea\label{eq:Theorem1}
~~
\eea
for both the joint and separate encoding schemes with multiplexing gain $r>0$.}
\end{Theorem}
\begin{IEEEproof}
As it is clear from (\ref{eq:P_JE}) and (\ref{eq:P_SE}) that $P^{\text{JE}}_{\text{out}}\leq P^{\text{SE}}_{\text{out}}$ for any linear transceiver $\Qv$ and $\Wv$,
we prove the theorem by showing that the upper-bound of $P^{\text{SE}}_{\text{out}}$ yields the same outage exponent as the lower-bound of $P^{\text{JE}}_{\text{out}}$.
Note that with the MMSE strategy, the $k$-th output SNR equals $\tau_k=\rho/[\Rv_e]_{k,k}-1$ \cite{Palomar:03} and
the target rate $R$ is set to be $R(\rho)=r\log\rho$ with $r>0$.

{\textbf{(1) DMT Lower-bound}}:
When the separate encoding is concerned with the MMSE transceiver, the outage probability (\ref{eq:P_SE}) is equivalently
\bea\label{eq:P_SE1}
P^{\text{SE}}_{\text{out}}&\doteq&\max_{i}P\left(\frac{1}{2}\log\left(\frac{\rho}{[\Rv_e]_{i,i}}\right)<\frac{R(\rho)}{N_S}\right).
\eea
Then, applying $\hat{\Bv}$ described in (\ref{eq:optimal relay matrix}) to the error covariance matrix (\ref{eq:lemma2}),
we obtain
\bea\label{eq:P_SE2}
P^{\text{SE}}_{\text{out}}&\doteq&P\left(-\log\left(\min_i(S_{h,i}+S_{g,i})\right)<\frac{2R(\rho)}{N_S}\right),
\eea
where
\bea\label{eq:S_hi}
S_{h,i}&\triangleq&[\Uv_h(\rho\mathbf{\Lambda}_h+\Iv_{N_S})^{-1}\Uv_h^H]_{i,i}\nonumber\\
&=&\uv_i(\rho\mathbf{\Lambda}_h+\Iv_{N_S})^{-1}\uv_i^H\nonumber\\
&=&\sum_{k=1}^{N_S}\frac{|u_{k,i}|^2}{1+\rho\lambda_{h,k}}
\eea
with $\uv_i$ being the $i$-th row of the unitary matrix $\Uv_h$ and $u_{k,i}$ being the $k$-th element of this column.
Similarly, we obtain
\bea\label{eq:S_gi}
S_{g,i}&\triangleq&\widetilde{\Uv}_h(\rho\hat{\mathbf{\Phi}}\widetilde{\mathbf{\Lambda}}_g\hat{\mathbf{\Phi}}+\rho\widetilde{\mathbf{\Lambda}}_y)^{-1}\widetilde{\Uv}_h^H]_{i,i}\nonumber\\
&=&\sum_{k=1}^{M}\frac{|u_{k,i}|^2}{|\hat{\phi}_k|^2\rho\lambda_{g,k}+\rho\lambda_{y,k}^{-1}}.
\eea
As it is always true that $|u_{k,i}|^2\leq1$ for all $(k,i)$, it simply follows from (\ref{eq:P_SE2}) that
\bea
P^{\text{SE}}_{\text{out}}~~~~~~~~~~~~~~~~~~~~~~~~~~~~~~~~~~~~~~~~~~~~~~~~~~~~~~~~~~~~~~~~~\nonumber\\
\dot{\leq}P\bigg(\sum_{k=1}^{N_S}\frac{1}{1+\rho\lambda_{h,k}}
+\sum_{k=1}^{M}\frac{1}{|\hat{\phi}_k|^2\rho\lambda_{g,k}+\rho\lambda_{y,k}^{-1}}>2^{-\frac{2R(\rho)}{N_S}}\bigg)\nonumber\\
\leq P\bigg(\sum_{k=1}^{N_S}\frac{1}{1+\rho\lambda_{h,k}}
+\sum_{k=1}^{M}\frac{1}{\eta\rho\lambda_{g,k}+\rho\lambda_{y,k}^{-1}}>2^{-\frac{2R(\rho)}{N_S}}\bigg),~~\nonumber
\eea
\vspace{-20pt}
\bea\label{eq:eta=1}
~~
\eea
where the second inequality holds because
$\hat{\mathbf{\Phi}}$ is optimum in terms of the trace minimization as shown in (\ref{eq:modified problem}) and we have
\bea\label{eq:trace change}
~~~\sum_{k=1}^{M}\frac{1}{|\hat{\phi}_k|^2\rho\lambda_{g,k}+\rho\lambda_{y,k}^{-1}}
=\text{Tr}\big(\rho\hat{\mathbf{\Phi}}^H\widetilde{\mathbf{\Lambda}}_g\hat{\mathbf{\Phi}}+\rho\widetilde{\mathbf{\Lambda}}_{y}^{-1}\big)^{-1}.\nonumber
\eea
Thus, setting $\hat{\mathbf{\Phi}}=\sqrt{\eta}\Iv_{M}$ clearly leads to the outage upper-bound (\ref{eq:eta=1}).
In addition, the following lemma, proven in Appendix \ref{Appendix:Proof of Lemma3}, implies that
$\text{Tr}(\Bv\Rv_y\Bv^H)<\text{Tr}(\rho\Bv\Bv^H)=\text{Tr}(\rho\hat{\mathbf{\Phi}}\hat{\mathbf{\Phi}}^H)$,
which means that $\eta$ can be chosen by $1$ to satisfy the relay power constraint.
\begin{Lemma}\label{Lemma:Lemma3}
{\it The covariance matrix of the relay signal $\yv$ (or the MMSE estimate of $\xv$ at the relay) in (\ref{eq:Ry})
is upper-bounded by the identity matrix as $\Rv_y\preceq\rho\Iv_{N_S}$.}
\end{Lemma}

Using the results presented above and setting $R(\rho)=r\log\rho$, we obtain
\bea
P^{\text{SE}}_{\text{out}}~~~~~~~~~~~~~~~~~~~~~~~~~~~~~~~~~~~~~~~~~~~~~~~~~~~~~~~~~~~~\nonumber\\
\label{eq:outage definition}
\dot{\leq} P\left(\sum_{k=1}^{N_S}\frac{1}{1+\rho\lambda_{h,k}}+\sum_{k=1}^{M}\frac{1}{\rho\lambda_{g,k}+\rho\lambda_{y,k}^{-1}}> \rho^{-\frac{2r}{N_S}}\right)~~\\
\label{eq:loose bound}
\leq P\left(\frac{1}{\rho\lambda_{h,N_S}}+\frac{1}{\rho\lambda_{g,M}}> \rho^{-\frac{2r}{N_S}}\right),~~~~~~~~~~~~~~~~~~~~~
\eea
Finally, applying the harmonic mean bounds, i.e., $\frac{\min(a,b)}{2}\leq\frac{ab}{a+b}\leq\min(a,b)$ for $a>0$ and $b>0$,
we have the trivial asymptotic upper-bound
\bea\label{eq:mu}
P_{\text{out}}^{\text{SE}}\dot{\leq} P\left(\mu<\rho^{-\left(1-\frac{2r}{N_S}\right)}\right)
\eea
with $\mu\triangleq \min(\lambda_{h,N_S},\lambda_{g,M})$. We notice that the upper-bound (\ref{eq:mu}) vanishes only when $2r/N_S<1$.
Hence, the outage exponent lower-bound equals zero for $2r/N_S\geq1$. Supposing $2r/N_S<1$, we can write
$P_{\text{out}}~\dot{\leq}~F_{\mu}\big(\rho^{-(1-\frac{2r}{N_S})^+}\big)$ where
$F_{\mu}(\cdot)$ stands for the cumulative distribution function of $\mu$.

Meanwhile, the result in \cite[Lemma 2]{Changick:11TWC} implies that
for a small argument $\delta\ll1$, $F_{\mu}(\delta)$ asymptotically equals
\bea\label{eq:F_mu}
F_{\mu}(\delta)\doteq F_{\lambda_{h,N_S}}(\delta)+F_{\lambda_{g,M}}(\delta)
\eea
where $F_{\lambda_{h,N_S}}(\delta)\propto\delta^{(N_R-N_S+1)^+}$ and $F_{\lambda_{g,M}}(\delta)\propto\delta^{(N_R-M+1)(N_D-M+1)^+}$ \cite{Sengul:06}.
Therefore, the resulting outage upper-bound is
\bea\label{eq:dME}
P_{\text{out}}^{\text{SE}}&\dot{\leq}&\rho^{(N_R-N_S+1)^+\min(1,(N_D-M+1)^+)\left(1-\frac{2r}{N_S}\right)^+}\nonumber\\
&=&\rho^{-d_{\text{ME}}(r)},
\eea
and we establish the DMT lower-bound of the sperate encoding scheme.

\textbf{ (2) DMT Upper-bound }:
In what follows, we examine the DMT upper-bound of the joint encoding scheme.
As $\log\left(\cdot\right)$ is a concave function,
$P_{\text{out}}^{\text{JE}}$ in (\ref{eq:P_JE}) is bounded by the Jensen's inequality as
\bea\label{eq:Jensen}
P_{\text{out}}^{\text{JE}}&\geq&P\left(\frac{N_S}{2}\log\left(\frac{1}{N_S}\sum_{i=1}^{N_S}\frac{\rho}{[\Rv_e]_{i,i}}\right)<R(\rho)\right)\nonumber\\
&=&P\left(\frac{1}{N_S}\sum_{i=1}^{N_S}\frac{1}{S_{h,i}+S_{g,i}}<\rho^{\frac{2r}{N_S}}\right)\nonumber\\
&\geq&P\left(\min_i(S_{h,i}+S_{g,i})>\rho^{-\frac{2r}{N_S}}\right),
\eea
where $S_{h,i}$ and $S_{g,i}$ are defined in (\ref{eq:S_hi}) and (\ref{eq:S_gi}), respectively.

Now, let us define $\bar{i}\triangleq\arg\min_i S_{h,i}+S_{g,i}$ and
let $\mathcal{A}$ be the event $\{|u_{k,\bar{i}}|^2\geq \frac{1-\epsilon}{N_S},~\forall k \}$
where $\epsilon>0$ is a small positive number independent of $\rho$.
Then, we can show that $P(\mathcal{A})$ is finite and independent of $\rho$, i.e, $P(\mathcal{A})\doteq\rho^0$ similar to \cite[Appendix A]{Kumar:09}.
Therefore, the outage probability can be further bounded by
\bea\label{eq:Si}
P_{\text{out}}^{\text{JE}}~\geq ~ P\bigg(\sum_{k=1}^{N_S}\frac{|u_{k,\bar{i}}|^2}{1+\rho\lambda_{h,k}}+\sum_{k=1}^{M}\frac{|u_{k,\bar{i}}|^2}{|\hat{\phi}_k|^2\rho\lambda_{g,k}+\rho\lambda_{y,k}^{-1}}\nonumber\\
>~\rho^{-\frac{2r}{N_S}}\Bigg|\mathcal{A}\bigg)P(\mathcal{A})\nonumber\\
\dot{\geq}~~ P\bigg(\sum_{k=1}^{N_S}\frac{1}{1+\rho\lambda_{h,k}}+\sum_{k=1}^{M}\frac{1}{|\hat{\phi}_k|^2\rho\lambda_{g,k}+\rho\lambda_{y,k}^{-1}}\nonumber\\
>~\frac{N_S}{1-\epsilon}\rho^{-\frac{2r}{N_S}}\bigg)\nonumber\\
~\geq~P\bigg(\sum_{k=1}^{N_S}\frac{1}{1+\rho\lambda_{h,k}}+\sum_{k=1}^{M}\frac{1}{N_S\rho\lambda_{y,k}^{-1}(1+\rho\lambda_{g,k})}\nonumber\\
>\frac{N_S}{1-\epsilon}\rho^{-\frac{2r}{N_S}}\bigg),
\eea
where the last inequality holds from the constraint $\text{Tr}(\Bv\Rv_y\Bv^H)\leq N_S\rho$ which implies that
$\hat{\mathbf{\Phi}}\widetilde{\mathbf{\Lambda}}_y\hat{\mathbf{\Phi}}^H\preceq \rho N_S\Iv_M$; thus,
$|\hat{\phi}_k|^2\leq N_S\rho\lambda_{y,k}^{-1}$ for all $k=1,\ldots,M$.
Note that from the definition of $\Rv_y$ (see (\ref{eq:Ry2}) in Appendix \ref{Appendix:Proof of Lemma3}),
we have $\rho\lambda_{y,k}^{-1}=1+\rho^{-1}\lambda_{h,k}^{-1}$ for all $k$.
Recalling that the multiplexing gain $r>0$ is assumed to be positive, the right-hand side of (\ref{eq:Si}) vanishes, and thus
the scaling factor $N_S/(1-\epsilon)$ does not affect the diversity order.
In other words, the outage lower-bound is equivalently
\bea\label{eq:tight r=0}
P_{\text{out}}^{\text{JE}}
\dot{\geq} P\Big(\frac{1}{1+\rho\lambda_{h,N_S}}~~~~~~~~~~~~~~~~~~~~~~~~~~~~~~~~~~~~~~\nonumber\\
+~~\frac{1}{(1+\rho^{-1}\lambda_{h,M}^{-1})(1+\rho\lambda_{g,M})}>\rho^{-\frac{2r}{N_S}}\Big).
\eea

Now, let us define
\bea\label{eq:alpha beta}
\alpha_{i}\triangleq-\frac{\log\lambda_{h,i}}{\log\rho}~~\text{and}~~\beta_j\triangleq-\frac{\log\lambda_{g,j}}{\log\rho},
\eea
for $i=1,\ldots,N_S$ and $j=1,\ldots,M$.
In addition, we define a positive real number
$0<\kappa<1$ as $\kappa\triangleq1-2r/N_S$ to make the outage expression more compact. Then, (\ref{eq:tight r=0}) is alternatively expressed as
\bea
P_{\text{out}}^{\text{JE}}~~~~~~~~~~~~~~~~~~~~~~~~~~~~~~~~~~~~~~~~~~~~~~~~~~~~~~~~~~~~~~~~~~~~\nonumber\\
\dot{\geq} P\bigg(\frac{1}{1+\rho^{1-\alpha_{N_S}}}+\frac{1}{(1+\rho^{\alpha_{N_S}-1})(1+\rho^{1-\beta_M})}>\rho^{-\frac{2r}{N_S}}\bigg)~~~\nonumber\\
\doteq P\bigg(\frac{1}{\rho^{\kappa-\alpha_{N_S}}}
\!+\!\frac{1}{\rho^{\alpha_{N_S}+\kappa-2}+\rho^{\kappa-\beta_M}+\rho^{\kappa-\beta_M+\alpha_{N_S}-1}}\!>\!1\!\bigg)~~\nonumber\\
\doteq P\bigg(\frac{1}{\rho^{\kappa-\alpha_M}}
+\frac{1}{\rho^{\alpha_{N_S}+\kappa-2}+\rho^{\max(\kappa,\kappa+\alpha_{N_S}-1)-\beta_M}}>1~~~~~~\nonumber\\
\Big|\alpha_{N_S}>\kappa\bigg)P\big(\alpha_{N_S}>\kappa\big)~~~\nonumber\\
+P\bigg(\frac{1}{\rho^{\kappa-\alpha_{N_S}}}+\frac{1}{\rho^{\alpha_{N_S}+\kappa-2}+\rho^{\max(\kappa,\kappa+\alpha_{N_S}-1)-\beta_M}}>1~~~\nonumber\\
\Big|\alpha_{N_S}<\kappa\bigg)P\big(\alpha_{N_S}<\kappa\big)~~~\nonumber\\
\overset{(a)}{\doteq} P\big(\alpha_{N_S}>\kappa\big)+P\big(\beta_M>\kappa\big)P\big(\alpha_{N_S}<\kappa\big)~~~~~~~~~~~~~~~~~~~~\nonumber
\eea
\vspace{-20pt}
\bea\label{eq:alpha NS}
\overset{(b)}{\doteq} F_{\lambda_{h,N_S}}(\rho^{-\kappa})+F_{\lambda_{g,M}}(\rho^{-\kappa}),~~~~~~~~~~~~~~~~~~~~~~~~~~~~
\eea
where (a) is due to the following exponential equalities:
\bea
\frac{1}{\rho^{\kappa-\alpha_{N_S}}}\doteq\left\{\begin{array}{cc} \infty & \text{if}~~\alpha_{N_S}>\kappa \\
0 & \text{if}~~\alpha_{N_S}<\kappa \end{array}\right.~\text{and}~~~~~~~~~~~~~~~~~~~~~\nonumber\\
\frac{1}{\rho^{\alpha_{N_S}-(2-\kappa)}+\rho^{\max(\kappa,\kappa+\alpha_{N_S}-1)-\beta_M}}\doteq~~~~~~~~~~~~~~~~~~~~\nonumber\\
\left\{\begin{array}{cc} \infty & \text{if}~~\alpha_{N_S}<2-\kappa~~\text{and}~ \beta_M>\max(\kappa,\kappa+\alpha_{N_S}-1) \\
0 & \text{if}~~\alpha_{N_S}>2-\kappa~~\text{or}~ \beta_M<\max(\kappa,\kappa+\alpha_{N_S}-1) \end{array}\right.\nonumber
\eea
and (b) follows from the fact that $P(\alpha_{k}<\kappa)\doteq\rho^0$ for any $k$ \cite{Zheng:03}.
Note that in an asymptotic sense with $\rho\rightarrow\infty$,
the probability of $\alpha_{k}$ or $\beta_{k}$ taking any value on the discontinuity point, e.g., $\alpha_{N_S}=\kappa$, is negligible \cite{Kumar:09}.
Now, it is immediate to check that (\ref{eq:alpha NS}) is asymptotically equivalent to (\ref{eq:mu}).
Therefore, applying the same argument as in (\ref{eq:F_mu}), we find out that
the DMT upper-bound coincides with the previously found lower-bound,
and the proof is concluded.
\end{IEEEproof}

While the ML-based transceiving scheme \cite{Tang:07} enjoys the optimal DMT (\ref{eq:optimal DMT}), Theorem \ref{Theorem:Theorem1} shows that
the linear transceiving scheme suffers from a significant diversity loss.
This is attributed to the fact that the MMSE transceiver enforces the transmitted
symbols to be spatially separated at the destination at the cost of the diversity order.
It is also observed that under the linear transceiving strategy,
there may be no advantage in coding across antennas in terms of the DMT compared to the separate encoding scheme.
This is because the output SNRs of virtual parallel channels,
i.e., $\tau_k$'s become strongly correlated with each other; thus,
only the minimum eigenvalue in each hop essentially dominates the performance as in (\ref{eq:tight r=0}).

Meanwhile, the DMT expression (\ref{eq:Theorem1}) reveals that
as long as $N_S\leq N_D$, the DMT is determined by the first-hop channel only.
This observation suggests that designing the system such that $N_S<N_D$ may not be an efficient, since
putting additional antennas at the destination over $N_S$ does not yield any DMT advantage.
It is also of interest to compare the MMSE transceiver in $N_S\times N_R\times N_D$ relay channels with
the MMSE receiver in $N_S\times N_D$ P2P channels.
From \cite[Theorem 1]{Kumar:09},
it is immediate to check that if we deploy a relay node ($N_R$) between the transmitter ($N_S$) and the receiver ($N_D$)
such that $N_R>N_D$, a higher diversity gain as well as the coverage extension can be achieved over the P2P systems,
although the multiplexing gain will be cut in half due to the half-duplex operation at the relay.

Finally, it is important to remark that Theorem \ref{Theorem:Theorem1} is valid only for the positive multiplexing gain $r>0$.
This is due to the limitation of the DMT framework as an asymptotic notion
which do not distinguish between different spectral efficiencies that correspond to the same multiplexing gain.
In fact, if $r=0$, the bound (\ref{eq:tight r=0}) which leads to the DMT upper-bound of the joint encoding scheme does not hold in general, because
each summation term in (\ref{eq:Si}) is bounded above as $1/(1+\rho\lambda_{h,k})<1$ and $1/(\rho\lambda_{y,k}^{-1}(1+\rho\lambda_{g,k}))<1$ for all $k$.
This means that other eigenvalues which do not appear in (\ref{eq:tight r=0}) may also play a role, and thus typically
leads to a higher diversity gain than (\ref{eq:Theorem1}).
Indeed, the MMSE transceiver exhibits ML-like performance
in case where the coding is applied across antennas with sufficiently low spectral efficiency.
Details will be addressed in Section \ref{sec:DRT} through the DRT analysis.

In the meantime, the ZF transceiver obtains the same DMT as one in (\ref{eq:Theorem1}) as shown in the following theorem, and thus
our statements so far can also be applied to the ZF transceiver.
However, it is important to note that unlike the MMSE in Theorem \ref{Theorem:Theorem1}, the DMT with the ZF transceiver holds for every multiplexing gain $r\geq0$.
Therefore, the joint and separate coding schemes exhibit completely the same diversity order.
\begin{Theorem}\label{Theorem:Theorem2}
{\it In $N_S\times N_R\times N_d$ MIMO AF half-duplex relaying channels, the DMT under the ZF transceiver is the same as $d_{\text{ME}}(r)$ in (\ref{eq:Theorem1})
for both the joint and separate encoding schemes and holds for all multiplexing gain $r\geq0$.}
\end{Theorem}
\begin{IEEEproof}
Similar to the case of Theorem \ref{Theorem:Theorem1},
the proof will be made by showing that the upper-bound of $P_{\text{out}}^{\text{JE}}$
and the lower-bound of $P_{\text{out}}^{\text{JE}}$ are asymptotically equivalent.
With the ZF strategy, we notice that $N_S\leq\min(N_R,N_D)$ and $\tau_k=\rho/[\Rv_e]_{k,k}$ \cite{Palomar:03},
and set the target data rate by $R(\rho)=r\log\rho$ with $r\geq0$.
The following lemma, proven in Appendix \ref{Appendix:Proof of Lemma4}, will be useful during our derivations.
\begin{Lemma}\label{Lemma:Lemma4}
{\it The covariance matrix of the relay signal $\zv$ (or the ZF estimate of $\xv$ at the relay) in (\ref{eq:Rz})
is exponentially equivalent to the identity matrix as $\Rv_z\doteq\rho\Iv_{N_S}$.}
\end{Lemma}

{\textbf{(1) DMT Lower-bound}}:
From the results in Section \ref{sec:ZF Transceiver} and (\ref{eq:P_SE}),
the outage probability, constrained to use the separate encoding, is written by
\bea
P^{\text{SE}}_{\text{out}}\doteq\max_{i}P\bigg(\log\Big(1+~~~~~~~~~~~~~~~~~~~~~~~~~~~~~~~~~~~~~~~\nonumber\\
\frac{\rho}{\big[(\Hv^H\Hv)^{-1}\big]_{i,i}
+\big[(\Uv_h\hat{\mathbf{\Phi}}\widetilde{\mathbf{\Lambda}}_g\hat{\mathbf{\Phi}}^H\Uv_h^H)^{-1}\big]_{i,i}}\Big)<\frac{2R(\rho)}{N_S}\bigg)\nonumber\\
\leq P\bigg(\min_i\big([(\rho\Hv^H\Hv)^{-1}]_{i,i}+\big[(\rho\Uv_h\hat{\mathbf{\Phi}}\widetilde{\mathbf{\Lambda}}_g\hat{\mathbf{\Phi}}^H\Uv_h^H)^{-1}\big]_{i,i}\big)\nonumber\\
>\rho^{-\frac{2r}{N_S}}\bigg)\nonumber\\
\leq P\left(\text{Tr}\big((\rho\Hv^H\Hv)^{-1}\big)+\text{Tr}\big((\rho\hat{\mathbf{\Phi}}\widetilde{\mathbf{\Lambda}}_g\hat{\mathbf{\Phi}}^H)^{-1}\big)
>\rho^{-\frac{2r}{N_S}}\right).\nonumber
\eea
Meanwhile, it is revealed from Lemma \ref{Lemma:Lemma4} that
$\text{Tr}(\Bv\Rv_z\Bv^H)\doteq\text{Tr}(\rho\hat{\mathbf{\Phi}}\hat{\mathbf{\Phi}}^H)\leq N_S\rho$.
Therefore, by setting $\hat{\mathbf{\Phi}}=\Iv_{N_S}$, the outage probability can be further bounded by
\bea\label{eq:ZF loose bound}
P^{\text{SE}}_{\text{out}}&\dot{\leq}&P\left(\sum_{i=1}^{N_S}\frac{1}{\rho\lambda_{h,i}}+\sum_{i=1}^{N_S}\frac{1}{\rho\lambda_{g,i}}>\rho^{-\frac{2r}{N_S}}\right)\nonumber\\
&\doteq&P\left(\frac{1}{\rho\lambda_{h,N_S}}+\frac{1}{\rho\lambda_{g,N_S}}>\rho^{-\frac{2r}{N_S}}\right),\nonumber
\eea
which is equivalent to the previous result in (\ref{eq:loose bound}),
and thus the DMT lower-bound is established.

{\textbf{(2) DMT Upper-bound}}: Considering the output SNR of the ZF transceiver and applying the Jensen's inequality,
 the outage probability (\ref{eq:P_JE}) is lower-bounded by
\bea
P_{\text{out}}^{\text{JE}}\geq~P\bigg(\frac{N_S}{2}\log\bigg(\frac{1}{N_S}\sum_{k=1}^{N_S}\Big(1+\frac{\rho}{[\Rv_e]_{k,k}}\Big)\bigg)<R(\rho)\bigg)\nonumber\\
\dot{\geq}~P\bigg(\min_i\Big([(\rho\Uv_h\mathbf{\Lambda}_h\Uv_h^H)^{-1}]_{i,i}~~~~~~~~~~~~~~~~~~~~~~~\nonumber\\
+\big[(\rho\Uv_h\hat{\mathbf{\Phi}}\widetilde{\mathbf{\Lambda}}_g\hat{\mathbf{\Phi}}^H\Uv_h^H)^{-1}\big]_{i,i}\Big)>\rho^{-\frac{2r}{N_S}}\bigg).\nonumber
\eea
According to Lemma \ref{Lemma:Lemma4}, the relay power constraint is asymptotically $\text{Tr}(\rho\hat{\mathbf{\Phi}}\hat{\mathbf{\Phi}}^H)\leq N_S\rho$;
thus, we have $\hat{\mathbf{\Phi}}\hat{\mathbf{\Phi}}^H\preceq N_S\Iv_{N_S}$, i.e., $|\hat{\phi}_k|^2\leq N_S$ for all $k$.
Then, applying the same argument as in (\ref{eq:Si}), we obtain
\bea
\!\!\!\!\!\!\!\!\!\!\!&&P_{\text{out}}^{\text{JE}}\nonumber\\
\!\!\!\!\!\!\!\!\!\!\!&&\dot{\geq} P\left(\min_i\left(\sum_{k=1}^{N_S}\frac{|u_{k,i}|^2}{\rho\lambda_{h,k}}
+\sum_{k=1}^{N_S}\frac{|u_{k,i}|^2}{|\hat{\phi}_k|^2\rho\lambda_{g,k}}\right)>\rho^{-\frac{2r}{N_S}}\right)\nonumber\\
\label{eq:ZF DMT UB}
\!\!\!\!\!\!\!\!\!\!\!&&\doteq P\left(\sum_{k=1}^{N_S}\frac{1}{\rho\lambda_{h,k}}
+\sum_{k=1}^{N_S}\frac{1}{N_S\rho\lambda_{g,k}}>\frac{N_S}{1-\epsilon}\rho^{-\frac{2r}{N_S}}\right)\\
\label{eq:ZF DMT UB2}
\!\!\!\!\!\!\!\!\!\!\!&&\doteq P\left(\frac{1}{\rho\lambda_{h,N_S}}
+\frac{1}{\rho\lambda_{g,N_S}}>\rho^{-\frac{2r}{N_S}}\right),
\eea
which exhibits the same outage expression as (\ref{eq:loose bound}); thus,
the DMT upper-bound is readily obtained by following similar steps from (\ref{eq:mu}) to (\ref{eq:dME}).
We note that in contrast to the previous result in (\ref{eq:Si}), each summation term in (\ref{eq:ZF DMT UB})
is unbounded above for small channel gains $\rho\lambda_{h,k}$ or $\rho\lambda_{g,k}$,
which implies that there is no room for the eigenvalues larger than
$\lambda_{h,N_S}$ and $\lambda_{g,N_S}$ to contribute to the outage probability in (\ref{eq:ZF DMT UB})
no matter which coding scheme is applied with finite rate ($r=0$).
Therefore, the derived DMT holds for all multiplexing gain $r\geq0$.
\end{IEEEproof}

In what follows, we study the DMT performance of the naive-ZF and -MMSE schemes to examine the effect of no CSI at the relay.
\begin{Theorem}\label{Theorem:Theorem3}
{\it The DMT of the $N_S\times N_R\times N_d$ MIMO AF half-duplex relaying channels with the naive-MMSE is given by
\bea\label{eq:Theorem3}
d_{\text{N-ME}}(r)=(\min(N_R,N_D)-N_S+1)^+\left(1-\frac{2r}{N_S}\right)^+
\eea
for both the joint and separate encoding schemes with positive multiplexing gain $r>0$.}
\end{Theorem}
\begin{IEEEproof}
Let us assume that $\Qv=\delta\Iv_{N_R}$ where $\delta$ is chosen to
satisfy the relay power constraint (\ref{eq:relay power constraint}) as
\bea
\delta^2=\frac{P_R}{\text{Tr}(\rho\Hv\Hv^H+\Iv_{N_R})}=\left(\frac{1}{N_S}\sum_{k=1}^{N_S}\lambda_{h,k}+\rho^{-1}\right)^{-1}.
\eea
Then, it is readily seen that $\delta\doteq c$ for some real positive value $c$ because we have
$\frac{1}{N_S}\sum_{k=1}^{N_S}\rho^{-\alpha_k}\doteq\rho^0$, i.e.,
the variable gain $\delta$ based on the channel state $\Hv$ is exponentially equivalent to the fixed gain relay $\delta=c$.
Similarly, one can show that $\Iv_{N_D}\succeq\Rv_n^{-1}\succeq(1+\delta^2\lambda_{g,1})^{-1}\Iv_{N_D}\doteq\rho^0\Iv_{N_D}$ \cite{SYang:07},
which means that the amplified noise at the relay does not affect the diversity order; thus,
the naive relaying can be regarded as a Rayleigh product channel \cite{SYang:11} whose error covariance matrix is
given by $\Rv_e=(c\Hv^H\Gv^H\Gv\Hv+\rho^{-1}\Iv_{N_S})^{-1}$.

Keeping this in mind, let us focus on the outage lower-bound of the joint encoding scheme.
Define $\lambda_{t,1}>\cdots>\lambda_{t,N_S}$ as the eigenvalues of $\Hv^H\Gv^H\Gv\Hv$.
Then, setting the target rate $R(\rho)=r\log\rho$ with $r>0$ and following the similar approaches as in (\ref{eq:Jensen}) and (\ref{eq:Si}),
it is easy to show that
\bea\label{eq:naive LB}
P_{\text{out}}^{\text{JE}}&\dot{\geq}& P\left(\sum_{k=1}^{N_S} \frac{1}{1+\rho\lambda_{t,k}}>\frac{N_S}{1-\epsilon}\rho^{-\frac{2r}{N_S}}\right)\nonumber\\
&\doteq&P\left(\frac{1}{\rho\lambda_{t,N_S}}>\rho^{-\frac{2r}{N_S}}\right).
\eea
First, we observe that if $N_S>\min(N_R,N_D)$, the outage probability (\ref{eq:naive LB}) leads to a trivial solution $P_{\text{out}}^{\text{JE}}~\dot{\leq}~\rho^0$,
due to the rank constraint of $\Hv^H\Gv^H\Gv\Hv$ which is equal to $N=\min(N_S,N_R,N_D)$. Thus, we assume $N_S\leq\min(N_R,N_D)$ from now on.
Note that the exponential equality (\ref{eq:naive LB}) holds only when $r>0$ for similar reason to (\ref{eq:tight r=0}).
Now, we define $\gamma_{k}\triangleq-\log\lambda_{t,k}/\log\rho$ for $k=1,\ldots,N_S$.
Then, (\ref{eq:naive LB}) is further simplified as $P_{\text{out}}^{\text{JE}}~\dot{\geq}~P(\mathcal{E}_{\gamma})\doteq\rho^{-d(r)}$ with the outage event
$\mathcal{E}_{\gamma}\triangleq\{\gamma_{N_S}>(1-\frac{2r}{N_S})^+\}$.

Let $f(\cv)$ be the p.d.f. of a random vector $\cv=[\gamma_1,\ldots,\gamma_{N_S}]$ and
$\theta(\cv)$ denote its exponential order, i.e., $f(\cv)\doteq\rho^{-\theta(\cv)}$. Then, one can show that
the DMT is calculated as \cite{Zheng:03}
\bea
d(r)=\inf_{\mathbf{c}\in\mathcal{E}_{\gamma},\forall\gamma_k>0}\theta(\cv).\nonumber
\eea
For Rayleigh product channels, it was shown in \cite{SYang:11} that $\theta(\cv)$ is given in three different forms according to antenna configurations\footnote{
For completeness, some of key results of \cite{SYang:11} are summarized in Appendix \ref{Appendix:Rayleigh Product}.}, but
we only need to consider two cases $N_R\leq N_D$ and $N_D< N_R$ since we assume $N_S=N$.
For the first case, by applying the result (\ref{eq:fc2}) in Appendix \ref{Appendix:Rayleigh Product}, we have
\bea
d(r)&\doteq&\inf_{\mathbf{c}\in\mathcal{E}_{\gamma},\forall\gamma_k>0}\theta_2(\cv)\nonumber\\
&=&(N_R-N_S+1)\left(1-\frac{2r}{N_S}\right)^+,\nonumber
\eea
as the infimum is obtained when $\gamma_{k}=0$ for $k=1,\ldots,N_S-1$ and $\gamma_{N_S}=1$.
Similarly, for the second case, by adopting the result in (\ref{eq:fc3}), we have
\bea
d(r)&\doteq&\inf_{\mathbf{c}\in\mathcal{E}_{\gamma},\forall\gamma_k>0}\theta_3(\cv)\nonumber\\
&=&(N_D-N_S+1)\left(1-\frac{2r}{N_S}\right)^+.\nonumber
\eea
Finally, combining of the two, we prove that the DMT upper-bound equals $d_{\text{N-ME}}(r)$ in (\ref{eq:Theorem3}).
Meanwhile, it is immediate to show that the separate encoding scheme achieves
the same DMT. Details are trivial, and thus omitted.
\end{IEEEproof}

On the other hand, with the naive-ZF, the error covariance matrix will be $\Rv_e=(c\Hv^H\Gv^H\Gv\Hv)^{-1}$ which leads to the outage lower-bound
\bea\label{eq:naive ZF LB1}
P_{\text{out}}^{\text{JE}}&\dot{\geq}& P\left(\sum_{k=1}^{N_S} \frac{1}{\rho\lambda_{t,k}}>\frac{N_S}{1-\epsilon}\rho^{-\frac{2r}{N_S}}\right).\\
\label{eq:naive ZF LB2}
&\doteq&P\left(\frac{1}{\rho\lambda_{t,N_S}}>\rho^{-\frac{2r}{N_S}}\right).
\eea
Recall that each summation term $1/\rho\lambda_{t,k}$ in (\ref{eq:naive ZF LB1}) is unbounded above for the small channel gain $\rho\lambda_{t,k}$.
Therefore, in contrast to (\ref{eq:naive LB}), the equality (\ref{eq:naive ZF LB2}) holds for every multiplexing gain $r\geq0$,
from which the DMT of the naive-ZF follows.
\begin{Theorem}\label{Theorem:Theorem4}
{\it In the $N_S\times N_R\times N_d$ MIMO AF half-duplex relaying channels, the DMT under the naive-ZF is the same as $d_\text{N-ME}(r)$
(\ref{eq:Theorem3}) for both the joint and separate encoding schemes and holds for all positive multiplexing gain $r\geq0$.}
\end{Theorem}

The DMT analysis of the naive schemes above provides useful insights on the AF relaying systems.
First, it is seen from Theorem \ref{Theorem:Theorem3} and \ref{Theorem:Theorem4}
that when the number of relay antennas is small such that $N_R\leq N_D$,
the naive schemes achieve the same DMT as the corresponding transceivers studied in Theorem \ref{Theorem:Theorem1} and \ref{Theorem:Theorem2}.
In this case, therefore, knowing the CSI at the relay seems to be insignificant.
However, as $N_R$ grows larger than $N_D$, we observe that the DMT of the naive schemes are always inferior to that of
the linear transceiving schemes.
This is due to the fact that the naive schemes do not fully exploit the transmit diversity
offered by the second-hop MIMO channel due to lack of the CSI at the relay\footnote{
It has been shown in \cite{SGharan:09} and \cite{RPedarsani:10} that the {\it ``random sequential''} scheme may also improve the DMT performance of the naive scheme
without needing the CSI at the relay. However, this scheme randomly changes the effective channel across the transmission block,
which amounts to the fast fading channels that require a continuous channel estimate at the destination; thus, is beyond the scope of this paper.}.
In other words, when $N_R>N_D$, a proper relay matrix design using the CSI is essential to obtain the DMT (\ref{eq:Theorem1}).

Similar to the MMSE transceiver, the DMT of the naive-MMSE only holds for positive multiplexing gain $r>0$.
In particular, when the rate is fixed and not too large, it is seen that very significant performance advantage is achieved by the joint encoding scheme.
Unlike the MMSE transceiver, however, we should remark that
the full-diversity order may not be achievable with the naive-MMSE scheme no matter how small the rate is.
This statement will be demonstrated through the DRT analysis in the subsequent section.

\section{Diversity-Rate Tradeoff Analysis with Joint Encoding}\label{sec:DRT}

In this section, we aim to characterize the diversity order of the MMSE schemes as a function of the finite spectral efficiency $R$ (b/s/Hz).
In fact, the DMT analysis accurately predicts the diversity order for the positive multiplexing gain $r>0$.
However, when the rate is fixed and small,
the MMSE with the joint encoding scheme exhibits the performance in stark contrast to one predicted by the DMT
(similar observation has been made in P2P channels \cite{Hedayat:07} \cite{Mehana:12} \cite{Kumar:09}).
In particular, the MMSE transceiver shows the full-diversity behavior when the rate is sufficiently low.
In what follows, we characterize the rate-dependent behavior of the MMSE schemes in AF relaying channels
through the DRT analysis which yields the tight upper and lower bounds of the diversity order as a function of the target rate $R$
and the number of antennas at each node\footnote{Note that tight bounds of the DRT with the separate encoding scheme are still open.
However, extensive computer simulations demonstrate that the MMSE schemes with
the separate encoding behaves as predicted by the DMT with $r=0$ for entire range of $R$}.

\begin{Theorem}\label{Theorem:Theorem5}
{\it For the fixed spectral efficiency $R$,
the DRT of the MMSE transceiver over $N_S\times N_R\times N_d$ MIMO AF half-duplex relaying channels,
constrained to use the joint encoding scheme, is given by
\bea\label{eq:Theorem5}
D_{\text{ME}}\left(\left\lceil (m)^+ \right\rceil\right)\leq d_{\text{ME}}(R)\leq D_{\text{ME}}\left(\left\lfloor \big(m+1)^+ \right\rfloor\right),
\eea
where $m\triangleq N_S 2^{-\frac{2R}{N_S}}+M-N_S$ and $D_{\text{ME}}(i)\triangleq\min\big(i(N_R+N_S-2M+i),(N_R-M+i)(N_D-M+i)^+\big)$,}
\end{Theorem}
\begin{IEEEproof}
{\textbf{(1) DRT Lower-bound}}:
Consider the MMSE transceiver with the joint encoding scheme. Since $-\log(\cdot)$ is convex,
applying the Jensen's inequality and setting the target rate as $R$, the outage probability of (\ref{eq:P_JE}) is upper-bounded by
\bea\label{eq:JE-MI-LB}
P_{\text{out}}^{\text{JE}}&\leq&P\left(-\frac{N_S}{2}\log\Big(\frac{1}{\rho N_S}\text{Tr}(\Rv_e)\Big)<R\right)\\
&\leq&P\bigg(-\frac{N_S}{2}\log\Big(\frac{1}{N_S}\big(\text{Tr}(\rho\mathbf{\Lambda}_h+\Iv_{N_S})^{-1}\nonumber\\
&&+~~\text{Tr}\big(\rho\widetilde{\mathbf{\Lambda}}_g+\rho\widetilde{\mathbf{\Lambda}}_{y}^{-1}\big)^{-1}\big)\Big)<R\bigg)\nonumber.
\eea
Note that the last bound is obtained by setting $\hat{\mathbf{\Phi}}=\Iv_{M}$ as in (\ref{eq:eta=1}).
Then, by employing the definitions of $\alpha_k$ and $\beta_k$ in (\ref{eq:alpha beta}), it follows
\bea
P_{\text{out}}^{\text{JE}}\leq P\bigg(\sum_{k=1}^{N_S}\frac{1}{1+\rho\lambda_{h,k}}+\sum_{k=1}^{M}\frac{1}{\rho\lambda_{g,k}+\rho\lambda_{y,k}^{-1}}>N_S 2^{-\frac{2R}{N_S}}\bigg)\nonumber\\
=P\bigg(\sum_{k=1}^{M}\frac{1}{1+\rho\lambda_{h,k}}+\sum_{k=1}^{M}\frac{1}{\rho\lambda_{g,k}+\rho\lambda_{y,k}^{-1}}>m\bigg)~~~~~~~~~~~~\nonumber\\
\doteq P\bigg(\sum_{k=1}^{M}\frac{1}{1+\rho^{1-\alpha_k}}+\sum_{k=1}^{M}\frac{1}{1+\rho^{1-\beta_k}+\rho^{\alpha_k-1}}> (m)^+\bigg),\nonumber
\eea
\vspace{-20pt}
\bea\label{eq:m^+}
~~
\eea
where (\ref{eq:m^+}) is due to the fact that the outage exponent will converge to $0$ for all $m\leq0$.

Asymptotically, the following exponential equalities hold:
\bea\label{eq:alpha_k}
\frac{1}{1+\rho^{1-\alpha_k}}\doteq\left\{\begin{array}{cc} 1 & \text{if}~\alpha_k>1 \\
0 & \text{if}~\alpha_k<1 \end{array}\right.~~~~~~~~~~~~~~~\\
\label{eq:beta_k}
\frac{1}{1+\rho^{1-\beta_k}+\rho^{\alpha_k-1}}\doteq\left\{\begin{array}{cc} 1 & \text{if}~\alpha_k<1~\text{and}~\beta_k>1 \\
0 & \text{if}~\alpha_k>1~\text{or}~\beta_k<1 \end{array}\right.,
\eea
for $k=1,\ldots,M$, which implies that in order for the outage to occur,
at least $\overline{m}\triangleq\lceil m^+\rceil$ number of terms in (\ref{eq:m^+}) should be $1$ among $2M$ summation terms.
It is important to note that (\ref{eq:alpha_k}) and (\ref{eq:beta_k}) cannot simultaneously be $1$ at the same $k$,
which is a key feature of the MMSE transceiver enabling us to achieve
the full diversity order for sufficiently small rate $R$.

%%%%%%%%%%%%%%%%%%%%%%%%%%%%%%%%%%%%%%%%%%%%%%%%%%%%%%%%%%%%%%%%%%Figure
\begin{figure*} [!htp]
\begin{center}
\includegraphics[width=6.3in]{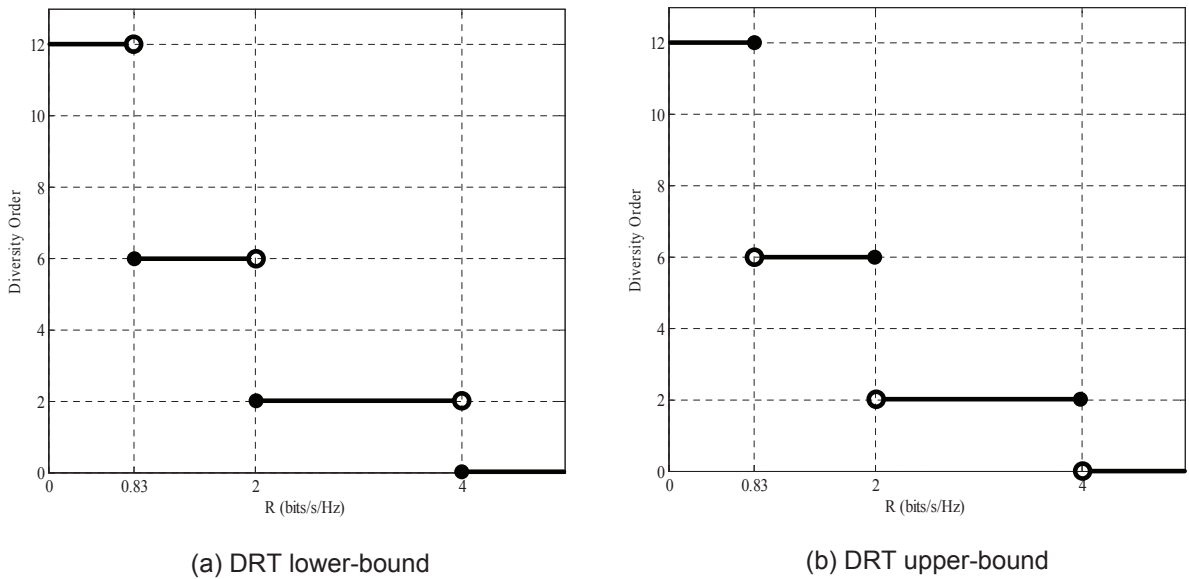}
\end{center}
\caption{DRT lower and upper bounds of the MMSE transceiver in $4\times4\times3$ MIMO AF relaying systems \label{figure:draw_diversity.eps}}
\end{figure*}
%%%%%%%%%%%%%%%%%%%%%%%%%%%%%%%%%%%%%%%%%%%%%%%%%%%%%%%%%%%%%%%%%%%%%%%%

Recall that all eigenvalues are arranged in descending order, which means that
$\{\alpha_i\}$ and $\{\beta_i\}$ are ordered according to $\alpha_1\leq\cdots\leq\alpha_M$ and $\beta_1\leq\cdots\leq\beta_M$.
For example, if $\alpha_1>1$, the term in (\ref{eq:beta_k}) converges to zero for all $k$, regardless of $\beta$.
Using this property, we can define the event $\mathcal{X}_i$ in which the summation in the left-hand side of (\ref{eq:m^+}) equals $i$ as
\bea
\mathcal{X}_i=\mathcal{E}_{h,i}\cup\mathcal{E}_{g,i,0}\cup\mathcal{E}_{g,i,1}\cup\cdots\cup\mathcal{E}_{g,i,i-1}\nonumber
\eea
for $i=1,\ldots,M$, where $\mathcal{E}_{h,i}\triangleq\!\big\{\alpha_{M-i+1}>1>\alpha_{M-i}\big\}$ and
$\mathcal{E}_{g,i,j}\triangleq\!\big\{\beta_{M-i+1}>1>\beta_{M-i}\big\}\cap\mathcal{E}_{h,j}$ for $j=0,1,\ldots,i-1$.
Then, from the union bound, we have
\bea\label{eq:Union to summation}
P_{\text{out}}^{\text{JE}}&\dot{\leq}&P\left(\bigcup_{i=\overline{m}}^M\mathcal{X}_i\right)\nonumber\\
&\leq&\sum_{i=\overline{m}}^M\Big(P(\mathcal{E}_{h,i})+\sum_{j=0}^{i-1}P(\mathcal{E}_{g,i,j})\Big),
\eea

First, we define $P(\mathcal{E}_{h,i})\doteq\rho^{-d_{h,i}(R)}$, $i=1,\ldots,M$. Then,
applying Varadhan's lemma \cite{Kumar:09} \cite{Zheng:03} by using the asymptotic p.d.f.\footnote{The p.d.f. is slightly different from \cite{Kumar:09}, since the eigenvalue ordering is reversed.}
of the random vector $\av=[\alpha_1,\ldots,\alpha_M]$ as
\bea\label{eq:pdf a}
f(\av)\doteq\Big[\prod_{l=1}^M\rho^{-(N_S+N_R-2l+1)\alpha_l}\Big]\text{exp}\Big(-\sum_{l=1}^M\rho^{-\alpha_l}\Big),
\eea
we obtain
\bea\label{eq:d_ih(R)}
d_{h,i}(R)&=&\inf_{\av\in\mathcal{E}_{h,i},\forall\alpha_l>0}\sum_{l=1}^{M}(N_S+N_R-2l+1)\alpha_l\nonumber\\
&=&\sum_{l=1}^{M-i}(N_S+N_R-2l+1)\times0\nonumber\\
&&+\sum_{l=M-i+1}^{M}(N_S+N_R-2l+1)\times1\nonumber\\
&=&i(N_R+N_S-2M+i).
\eea

Now, let us examine the probability of the event $\mathcal{E}_{g,i,j}$, i.e., $P(\mathcal{E}_{g,i,j})\doteq\rho^{-d_{g,i,j}(R)}$.
Defining $L\triangleq\min(N_R,N_D)$, the p.d.f. of the random vector $\bv=[\beta_1,\ldots,\beta_L]$ is given by
\bea\label{eq:pdf b}
f(\bv)\doteq\Big[\prod_{l=1}^L\rho^{-(N_R+N_D-2l+1)\beta_i}\Big]\text{exp}\Big(-\sum_{l=1}^L\rho^{-\beta_l}\Big).
\eea
Then, the probability of the event $\mathcal{E}_{g,i,j}$ is
\bea
P(\mathcal{E}_{g,i,j})=\int_{\mathcal{E}_{g,i,j}}f(\av,\bv)d\av d\bv~~~~~~~~~~~~~~~~~~~~~~~~~~~~~~~~\nonumber\\
\doteq\int_{\mathcal{E}_{g,i,j}}\left[\rho^{-\sum_{l=1}^M (N_S+N_R-2l+1)\alpha_l-\sum_{l=1}^L (N_R+N_D-2l+1)\beta_l}\right]\nonumber\\
\times\text{exp}\Big(-\sum_{l=1}^M\rho^{-\alpha_l}-\sum_{l=1}^L\rho^{-\beta_l}\Big)d\av d\bv,\nonumber
\eea
due to the independence of $\av$ and $\bv$, and applying Varadhan's lemma again, we have
\bea
d_{g,i,j}(R)~~~~~~~~~~~~~~~~~~~~~~~~~~~~~~~~~~~~~~~~~~~~~~~~~~~~~~~~~~~~\nonumber\\
=\inf_{(\mathbf{a},\mathbf{b})\in\mathcal{E}_{g,i,j},\forall\alpha_l,\forall\beta_l>0}\sum_{l=1}^{M}(N_S+N_R-2l+1)\alpha_l~~~~~~~~~~~~~\nonumber\\
+~\sum_{l=1}^{L}(N_R+N_D-2l+1)\beta_l~~~~\nonumber\\
=\!\!\!\!\!\!\sum_{l=M-j+1}^{M}\!\!(N_S+N_R-2l+1)+\!\!\!\!\!\!\sum_{l=M-i+1}^{L}(N_R+N_D-2l+1)\nonumber\\
=j(N_S+N_R-2M+j)~~~~~~~~~~~~~~~~~~~~~~~~~~~~~~~~~~~~~~~\nonumber\\
+~~(N_R+N_D-L-M+i)(L-M+i)^+~~\nonumber\\
=j(N_S+N_R-2M+j)\!+\!(N_R-M+i)(N_D-M+i)^+.\nonumber
\eea
\vspace{-20pt}
\bea\label{eq:d_ig(R)}
~~
\eea
Finally, we observe from (\ref{eq:d_ih(R)}) and (\ref{eq:d_ig(R)}) that
all outage events in (\ref{eq:Union to summation}) yield higher outage exponents
than $\mathcal{E}_{h,\overline{m}}$ or $\mathcal{E}_{g,\overline{m},0}$.
Therefore, we eventually conclude that
\bea
~~P_{\text{out}}^{\text{JE}}~\dot{\leq}~ P(\mathcal{E}_{h,\overline{m}})+P(\mathcal{E}_{g,\overline{m},0})\doteq\rho^{-\min(d_{h,\overline{m}}(R),d_{g,\overline{m},0}(R))},\nonumber
\eea
and the proof of DRT lower-bound is established.

{\textbf{ (2) DRT Upper-bound}}:
We start from the lower-bound of $P_{\text{out}}^{\text{JE}}$ defined in (\ref{eq:Si}). For a fixed rate $R$,
it can be rephrased by
\bea
P_{\text{out}}^{\text{JE}}~\dot{\geq}~ P\bigg(\sum_{k=1}^{M}\frac{1}{1+\rho\lambda_{h,k}}
+\sum_{k=1}^{M}\frac{1}{\rho\lambda_{y,k}^{-1}(1+\rho\lambda_{g,k})}>m_{\epsilon}\bigg)\nonumber\\
\doteq P\bigg(\sum_{k=1}^{M}\frac{1}{1+\rho^{1-\alpha_k}}~~~~~~~~~~~~~~~~~~~~~~~~~~~~~~~~~~~~~~\nonumber\\
~~+~\sum_{k=1}^{M}\frac{1}{\rho^{1-\beta_k}+\rho^{-(1-\alpha_k)}+\rho^{\alpha_k-\beta_k}}>(m_{\epsilon})^+\bigg),\nonumber
\eea
\vspace{-20pt}
\bea\label{eq:1-epsilon}
~~
\eea
where $m_{\epsilon}\triangleq\frac{N_S}{1-\epsilon}2^{-\frac{2R}{N_S}}+M-N_S$.
Asymptotically, the following equality holds:
\bea
\frac{1}{\rho^{1-\beta_k}+\rho^{-(1-\alpha_k)}+\rho^{\alpha_k-\beta_k}}\doteq\left\{\begin{array}{cc}\!\! 1 &\!\! \text{if}~\alpha_k<1~\text{and}~\beta_k>1 \\
\!\!0 & \!\!\text{if}~\alpha_k>1~\text{or}~\beta_k<1 \end{array}\right.\nonumber
\eea
\vspace{-15pt}
\bea\label{eq:alpha_k-beta_k}
~~
\eea
for $k=1,\ldots,M$, which is exponentially equivalent to (\ref{eq:beta_k}).
Therefore, the remaining proof simply follows the previously studied DRT lower-bound by replacing $m$ with $m_{\epsilon}$.
Finally, we have
\bea
P_{\text{out}}^{\text{JE}}~\dot{\geq}~ P(\mathcal{E}_{h,\overline{m}_{\epsilon}})+P(\mathcal{E}_{g,\overline{m}_{\epsilon},0})\doteq\rho^{-\min(d_{h,\overline{m}_{\epsilon}}(R),d_{g,\overline{m}_{\epsilon},0}(R))},\nonumber
\eea
where $\overline{m}_{\epsilon}\triangleq\lceil(m_{\epsilon})^+\rceil$.
As long as $N_S2^{-\frac{2R}{N_S}}\notin\mathbb{N}$ is non-integer, the constant $\epsilon$ can be chosen such that
$\overline{m}_{\epsilon}=\overline{m}$. Therefore, the upper and lower bounds are tight.
However, when $N_S2^{-\frac{2R}{N_S}}\in\mathbb{N}$ takes an integer value,
the outage exponent obeys a slightly weaker upper-bound with $\overline{m}_{\epsilon}=\lfloor(m+1)^+\rfloor$, and the proof is concluded.
\end{IEEEproof}

Our result in Theorem \ref{Theorem:Theorem5} confirms and complements the earlier work on DMT in Theorem \ref{Theorem:Theorem1}.
We first see that when the rate is high, i.e., $R>\frac{N_S}{2}\log N_S$ or $\lceil(m)^+\rceil=\lfloor(m+1)^+\rfloor=1$,
both Theorem \ref{Theorem:Theorem1} and \ref{Theorem:Theorem5} yield the same diversity.
At high rate, therefore, the diversity order of the MMSE transceivers may be predictable by DMT analysis with setting $r=0$,
and thus very suboptimal compared to the ML diversity (\ref{eq:optimal DMT}).
However, as the rate becomes lower,
it is shown from Theorem \ref{Theorem:Theorem2} that higher diversity order is actually achievable than one predicted by the DMT.
In particular, when $R<\frac{N_S}{2}\log\frac{N_S}{N_S-1}$ or $\lceil(m)^+\rceil=\lfloor(m+1)^+\rfloor=M$,
the MMSE transceivers even exhibit the ML-like performance with full diversity order $d^*(r=0)=N_R\min(N_S,N_D)$.
It is especially interesting to observe that when the rate is sufficiently small,
a certain diversity gain is still achievable even when $N_S>\min(N_R,N_D)$,
which is often overlooked in conventional works for MMSE-based MIMO relaying systems.

A careful examination of the bounds in (\ref{eq:Theorem5}) reveals that the upper-bound is left-continuous while the lower-bound is right-continuous at the
discontinuity points. To help readers understand better, we take an example in Figure \ref{figure:draw_diversity.eps} which
shows the DRT performance of the MMSE transceiver in $4\times4\times3$ MIMO AF relaying channels.
As seen, two bounds in (\ref{eq:Theorem5}) are very tight against each other except its discrepant points.
It is also confirmed from the figure that various diversity gains, up to the full diversity order, are achievable by adjusting the transmit rate.
As shown in the following, however, this may not be the case when the naive-MMSE scheme is adopted.

\begin{Theorem}\label{Theorem:Theorem6}
{\it Define $(X,Y,Z)$ be the ordered version of $(N_S,N_R,N_D)$ with $N\leq Y\leq Z$.
Then, for the fixed spectral efficiency $R$,
the DRT of the naive-MMSE over $N_S\times N_R\times N_d$ MIMO AF half-duplex relaying channels,
constrained to use the joint encoding scheme, is given by
\bea\label{eq:Theorem6}
D_{\text{N-ME}}\left(\left\lceil (n)^+ \right\rceil\right)\leq d_{\text{N-ME}}(R)\leq D_{\text{N-ME}}\left(\left\lfloor \big(n+1)^+ \right\rfloor\right),
\eea
where $n\triangleq N_S 2^{-\frac{2R}{N_S}}+N-N_S$ and
\bea
D_{\text{N-ME}}(i)\triangleq i(Y-N+i)-\left\lfloor\frac{\big[(i-(Z-Y))^+\big]^2}{4}\right\rfloor.\nonumber
\eea
}
\end{Theorem}
\begin{IEEEproof}
We prove the theorem by developing the DRT lower-bound of the naive-MMSE with the joint encoding scheme.
As mentioned previously,
the naive relay channel is asymptotically approximated to the Rayleigh product channel.
Thus, applying the similar argument as in (\ref{eq:JE-MI-LB}) and (\ref{eq:m^+}), we can write the outage upper-bound as
\bea
P_{\text{out}}^{\text{JE}}~~~~~~~~~~~~~~~~~~~~~~~~~~~~~~~~~~~~~~~~~~~~~~~~~~~~~~~~~~~~~~~~~~\nonumber\\
\leq P\bigg(-\frac{N_S}{2}\log\Big(\frac{1}{N_S}\text{Tr}[(\rho\Hv^H\Gv^H\Gv\Hv+\Iv_{N_S})^{-1}]\Big)<R\bigg)\nonumber\\
\doteq P\left(\sum_{k=1}^{N}\frac{1}{1+\rho^{1-\gamma_k}}>(n)^+\right),~~~~~~~~~~~~~~~~~~~~~~~~~~~~~~~\nonumber
\eea
where $n\triangleq N_S 2^{-\frac{2R}{N_S}}+N-N_S$.
Let us define events $\mathcal{E}_i=\{\gamma_i>1>\gamma_{i+1}\}$ for all $i$.
Then, for large $\rho$, the following approximation holds:
\bea
~~P\left(\sum_{k=1}^{N}\frac{1}{1+\rho^{1-\gamma_k}}>(n)^+\right)
\approx\bigcup_{i=\overline{n}}^{N}\mathcal{E}_i\leq \sum_{i=\overline{n}}^N P(\mathcal{E}_i),\nonumber
\eea
where $\overline{n}\triangleq\lceil(n)^+\rceil$.

We now define $P(\mathcal{E}_i)\doteq\rho^{-d_i(R)}$ for $i=1,\ldots,N$. Then, using the pdf $f(\cv)\doteq\rho^{-\theta(\cv)}$ of a
random vector $\cv=[\gamma_1,\ldots,\gamma_N]$ given in Appendix \ref{Appendix:Rayleigh Product}, we obtain

\begin{itemize}
\item for $N_R=N$,
\bea
d_i(R)=\inf_{\mathbf{c}\in\mathcal{E}_i,\forall \gamma_k>0}\theta_1(\cv)~~~~~~~~~~~~~~~~~~~~~~~~~~~~~~~~~~~~~~\nonumber\\
=i(\min(N_S,N_D)-N_R+i)-\left\lfloor\frac{\big[(i-|N_S-N_D|)^+\big]^2}{4}\right\rfloor\nonumber
\eea

\item for $N_D\leq N_R\leq N_S$ or $N_S\leq N_R\leq N_D$,
\bea
d_i(R)&=&\inf_{\mathbf{c}\in\mathcal{E}_i,\forall \gamma_k>0}\theta_2(\cv)\nonumber\\
&=&i(N_R-\min(N_S,N_D)+i)\nonumber\\
&&~~-~\left\lfloor\frac{\big[(i-|\max(N_S,N_D)-N_R|)^+\big]^2}{4}\right\rfloor\nonumber
\eea

\item for $N_D\leq N_S< N_R$ or $N_S\leq N_D< N_R$,
\bea
d_i(R)&=&\inf_{\mathbf{c}\in\mathcal{E}_i,\forall \gamma_k>0}\theta_3(\cv)\nonumber\\
&=&i(\max(N_S,N_D)-\min(N_S,N_D)+i)\nonumber\\
&&~~-~\left\lfloor\frac{\big[(i-|N_R-\min(N_S,N_D)|)^+\big]^2}{4}\right\rfloor\nonumber.
\eea
\end{itemize}
Finally, combining the above three, we arrive at $d_i(R)=D_{\text{N-ME}}(i)$. Thus, we can eventually conclude that
$\sum_{i=\overline{n}}^N P(\mathcal{E}_i)\doteq\rho^{-d_{\overline{n}}(R)}$,
and the DRT lower-bound is established.
For the DRT upper-bound, the proof follows as an immediate corollary from the previous results in Theorem \ref{Theorem:Theorem5}; thus is omitted.
\end{IEEEproof}

From Theorem \ref{Theorem:Theorem6}, some remarks can be made about the DRT performance of the naive-MMSE.
First, the upper and lower bounds in (\ref{eq:Theorem6}) are tight except some inconsistent points.
In addition, we observe that the DRT of the naive-MMSE does not depend on the antenna configuration ($N_S,N_R,N_D$), but only depends on its ordered triple ($N,Y,Z$).
Since $D_{\text{N-ME}}(i)$ is an increasing function of $i$, the maximum achievable diversity of the naive-MMSE is given by
\bea
NY-\left\lfloor\frac{\big[(N-Z+Y)^+\big]^2}{4}\right\rfloor,\nonumber
\eea
when $\lceil(n)^+\rceil=\lfloor(n+1)^+\rfloor=N$, i.e., $R<\frac{N_S}{2}\log\left(\frac{N_S}{N_S-1}\right)$.
This result shows that in order for the naive-MMSE to achieve the full-diversity order of MIMO AF relaying channels,
at least following two conditions: $N_R\in\{N,Y\}$ and $N<Z-Y+2$ must be satisfied, and thus
the full-diversity order is not achievable in general with the naive-MMSE.
Meanwhile, when the rate is sufficiently high such that $R>\frac{N_S}{2}\log N_S$, i.e., $\lceil(n)^+\rceil=\lfloor(n+1)^+\rfloor=1$,
the diversity order is the same as one predictable by the DMT in Theorem \ref{Theorem:Theorem3}.

\section{Numerical Results } \label{sec:Numerical Results}

%%%%%%%%%%%%%%%%%%%%%%%%%%%%%%%%%%%%%%%%%%%%%%%%%%%%%%%%%%%%%%%%%Figure
\begin{figure}
\begin{center}
\includegraphics[width=3.7in, height=3.0in]{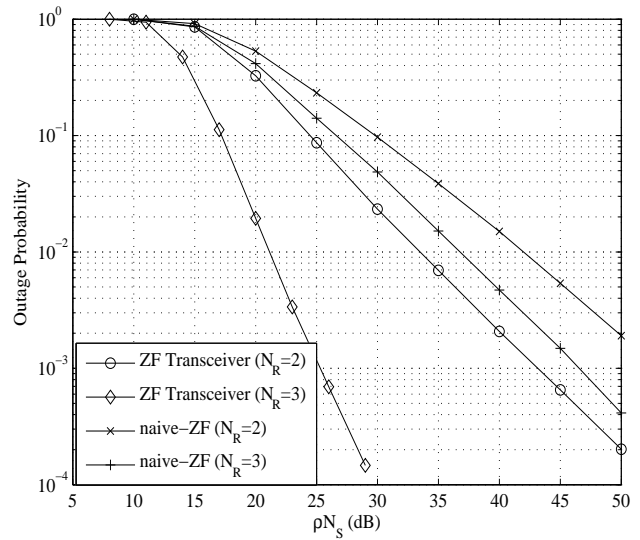}
\end{center}
\caption{Outage probabilities of ZF schemes under the joint encoding with $N_S=N_D=2$ and $R=3.32$ bits/s/Hz \label{figure:ZF2.eps}}
\end{figure}
%%%%%%%%%%%%%%%%%%%%%%%%%%%%%%%%%%%%%%%%%%%%%%%%%%%%%%%%%%%%%%%%%%%%%%
%%%%%%%%%%%%%%%%%%%%%%%%%%%%%%%%%%%%%%%%%%%%%%%%%%%%%%%%%%%%%%%%%Figure
\begin{figure}
\begin{center}
\includegraphics[width=3.7in, height=3.0in]{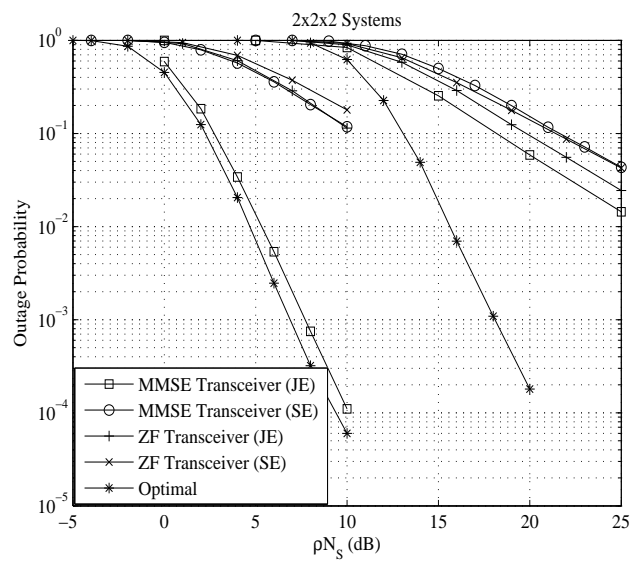}
\end{center}
\caption{Outage probabilities of MMSE and ZF transceivers with $R=0.42$ and $2$ bits/s/Hz (left to right) \label{figure:2x2x2MMSE2.eps}}
\end{figure}
%%%%%%%%%%%%%%%%%%%%%%%%%%%%%%%%%%%%%%%%%%%%%%%%%%%%%%%%%%%%%%%%%%%%%%

The goal of this section is to demonstrate the accuracy of our analysis and provide several interesting observations through numerical simulations.
For convenience, we denote the joint encoding and the separate encoding schemes by JE and SE, respectively.
Figure \ref{figure:ZF2.eps} compares two ZF schemes, i.e., the ZF transceiver and naive-ZF under the JE.
It is seen that the ZF transceiver not only obtains the power gain about $10$ dB over the naive-ZF,
but also achieves a diversity gain when $N_R>N_D$.
This result shows that a proper relay matrix design assisted by the CSI is indeed important to fully exploit the resources of the MIMO AF relaying systems.
We also confirm from this figure that our DMT analysis in Theorem \ref{Theorem:Theorem2} and \ref{Theorem:Theorem4} accurately predicts the numerical performance of the ZF schemes.

Figure \ref{figure:2x2x2MMSE2.eps} illustrates the case of $2\times2\times2$ systems with $R=0.42$ and $2$ bits/s/Hz.
Here, {\it ``Optimal''} exhibits the outage probability of the maximum MI described in (\ref{eq:max MI}).
As predicted by Theorem \ref{Theorem:Theorem5},
we observe that when the coding is applied jointly across antennas, the MMSE transceiver shows near optimal performance as the rate becomes smaller,
while the ZF transceiver exhibits parallel waterfall error curves regardless of the coding scheme and the code rate as predicted by the DMT in Theorem \ref{Theorem:Theorem2}.
Similar observation can be made in Figure \ref{figure:3x3x2MMSE.eps}
which compares the outage performance of the naive-MMSE and the MMSE transceiver in $3\times3\times2$ systems with $R=0.39$ bits/s/Hz.
It is shown that although the separate encoding schemes experience the outage floor due to lack of spatial dimension at the destination,
the MMSE schemes with joint encoding scheme enjoy a substantial performance advantage.
In fact, this observation is quite antithetic to the common assumption $N_S\leq\min(N_R,N_D)$
which has usually been adopted in MMSE-based MIMO relaying systems
\cite{Changick:09TWC} \cite{Ronghong:09} \cite{Tseng:10}.
It is also interesting to observe that unlike the MMSE transceiver,
the naive-MMSE does not achieve the full-diversity order, even if the rate is sufficiently small.
This is easily inferred from Theorem \ref{Theorem:Theorem5} and \ref{Theorem:Theorem6}
since the DRT of the naive MMSE exhibits only $d_{\text{N-ME}}(0.39)=5$ due to the penalty term in (\ref{eq:Theorem6}),
while the MMSE transceiver yields the full diversity $d_{\text{ME}}(0.39)=6$.

%%%%%%%%%%%%%%%%%%%%%%%%%%%%%%%%%%%%%%%%%%%%%%%%%%%%%%%%%%%%%%%%%Figure
\begin{figure}
\begin{center}
\includegraphics[width=3.7in, height=3.0in]{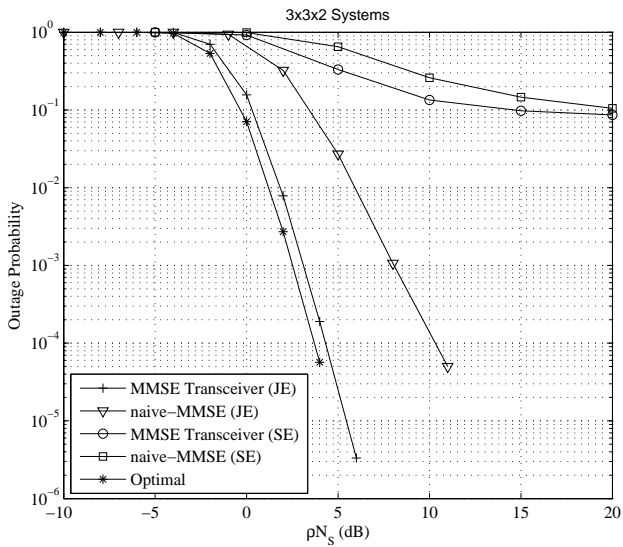}
\end{center}
\caption{Outage probabilities of the MMSE schemes with $R=0.39$ bits/s/Hz \label{figure:3x3x2MMSE.eps}}
\end{figure}
%%%%%%%%%%%%%%%%%%%%%%%%%%%%%%%%%%%%%%%%%%%%%%%%%%%%%%%%%%%%%%%%%%%%%%

Meanwhile, it is sometimes the case that adding antennas at each node may
be more convenient than insisting on high-complexity receiver processing at the destination \cite{Kumar:09}.
Let us consider the outage curves in Figure \ref{figure:OptimalvsMMSE.eps}.
Suppose that we want to achieve the rate $R=2$ bits/s/Hz at the block error rate $10^{-3}$ with SNR $\rho N_S=18$ dB.
Figure \ref{figure:OptimalvsMMSE.eps} shows that
this target performance is achieved by the $2\times2\times2$ optimal scheme, but obviously not via $2\times2\times2$ MMSE scheme.
One way to improve the performance of the MMSE scheme is to increase the number of antennas at the relay,
since we know that additional antennas at the relay leads to additional DMT advantage (this is not the case in the naive-MMSE).
A big merit of this method is that the target performance can be achieved even with the separate encoding scheme.
If the joint encoding is available, it may also be possible to improve the performance by increasing the data streams
because the rate becomes relatively small in a system with large $N_S$.
For example, it is shown that the MMSE transceiver in $3\times3\times3$ and $4\times4\times4$ systems
attains substantial performance gain over one in $2\times2\times2$ systems.

Finally, the outage probability of the MMSE transceiver is presented in Figure \ref{figure:334vs433.eps} for
$3\times3\times4$ and $4\times3\times3$ systems with various rates.
We see that although the performance of $4\times3\times3$ system may be poor and even floored at high rates since we have $N_S>N_D$,
it shows similar performance to its counter part in $N_S<N_D$ as the rate becomes smaller.
This result implies that when the rate is sufficiently small, increasing the data streams at the based-station
is as effective as increasing the receiver antennas at the destination.

%%%%%%%%%%%%%%%%%%%%%%%%%%%%%%%%%%%%%%%%%%%%%%%%%%%%%%%%%%%%%%%%%Figure
\begin{figure}
\begin{center}
\includegraphics[width=3.7in, height=3.0in]{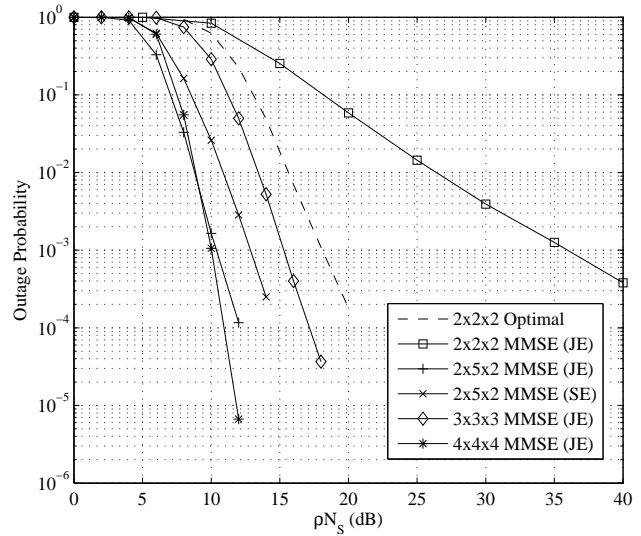}
\end{center}
\caption{Outage probabilities of the optimal and MMSE transceivers with $R=2$ bits/s/Hz \label{figure:OptimalvsMMSE.eps}}
\end{figure}
%%%%%%%%%%%%%%%%%%%%%%%%%%%%%%%%%%%%%%%%%%%%%%%%%%%%%%%%%%%%%%%%%%%%%%

%%%%%%%%%%%%%%%%%%%%%%%%%%%%%%%%%%%%%%%%%%%%%%%%%%%%%%%%%%%%%%%%%Figure
\begin{figure}
\begin{center}
\includegraphics[width=3.7in, height=3.0in]{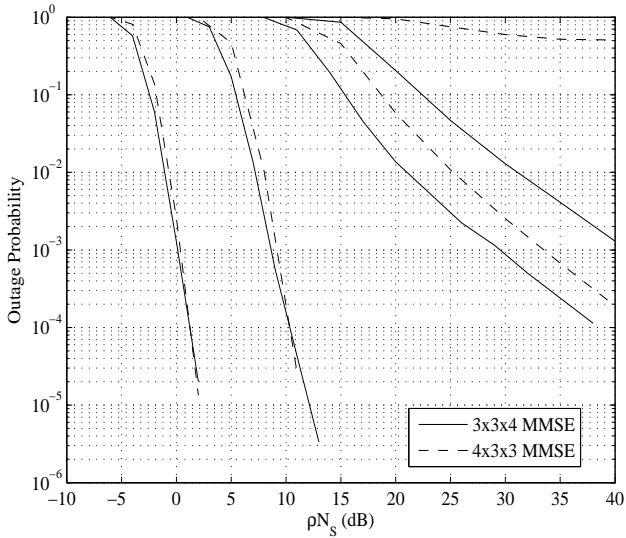}
\end{center}
\caption{Outage probabilities of the MMSE transceiver under the joint encoding with $R=0.3$, $1.2$, $3$, and $5$ bits/s/Hz (left to right) \label{figure:334vs433.eps}}
\end{figure}
%%%%%%%%%%%%%%%%%%%%%%%%%%%%%%%%%%%%%%%%%%%%%%%%%%%%%%%%%%%%%%%%%%%%%%

%%%%%%%%%%%%%%%%%%%%%%%%%%%%%%%%%%%%%%%%%%%%%%%%%%%%%%%%%%%%%%%%%%%%%%%%%%%%%%%%%%%%%%%%%%%%%%%%%%%%
%%%%%%%%%%%%%%%%%%%%%%%%%%%%%%%%%%%%%%%   Conclusion  %%%%%%%%%%%%%%%%%%%%%%%%%%%%%%%%%%%%%%%%%%%%%%
%%%%%%%%%%%%%%%%%%%%%%%%%%%%%%%%%%%%%%%%%%%%%%%%%%%%%%%%%%%%%%%%%%%%%%%%%%%%%%%%%%%%%%%%%%%%%%%%%%%%

\section{Conclusion} \label{sec:Conclusion}

In this paper, we provided comprehensive analysis on the diversity order of the linear transceivers in MIMO AF relaying systems.
We first presented a design framework for the transceiver optimization in terms of both the ZF and the MMSE, and then
provided two types of asymptotic performance analysis.
In the first part of the analysis, we studied the DMT of the ZF and MMSE transceivers.
While the DMT analysis accurately predicts their diversity performance for the positive multiplexing gain,
it was shown that the MMSE transceivers are very unpredictable via DMT when the rate is finite.
In the second part of the analysis, we highlighted this rate-dependent behavior of the
MMSE transceivers and characterized their diversity at all finite rates.
It is especially interesting to observe that the MMSE transceiver exhibits the ML-like performance as the rate becomes smaller,
but the full-diversity order is not guaranteed in the naive-MMSE, even if the rate is arbitrarily small.
Throughout our analysis on both the DMT and the DRT,
we offered complete understanding on the diversity order of the linear transceivers in MIMO AF relaying systems.
Finally, simulation results confirmed our analysis.

\appendices

\section{Proof of Lemma \ref{Lemma:Lemma1}} \label{Appendix:Proof of Lemma1}
In general, one can write $\Qv$ as $\Qv=\Qv_{\|}+\Qv_{\bot}$ \cite{SJang:10} where $\Qv_{\|}$ and $\Qv_{\bot}$ denote the components of
$\Qv$ such that the row space of $\Qv_{\|}$ and $\Qv_{\bot}$ are parallel and orthogonal to the column space of $\Hv$, respectively.
Then, from the definition of the error covariance matrix $\Rv_e$ which is unfolded as
\bea
\Rv_e(\Qv)\triangleq E[(\hat{\xv}-\xv)(\hat{\xv}-\xv)^H]~~~~~~~~~~~~~~~~~~~~~~~~~~~\nonumber\\
=\Wv(\rho\Gv\Qv\Hv\Hv^H\Qv^H\Gv^H+\Gv\Qv\Qv^H\Gv^H+\Iv_{N_D})\Wv^H\nonumber\\
+\rho\Wv\Gv\Qv\Hv+\rho\Hv^H\Qv^H\Gv^H\Wv^H+\rho\Iv_{N_S}\nonumber
\eea
and the relay power consumption $\text{Tr}(\Qv(\rho\Hv\Hv^H+\Iv_{N_R})\Qv^H)$, it is easy to see that
\bea
\Rv_e(\Qv_{\|})&\preceq&\Rv_e(\Qv)\nonumber\\
\text{Tr}(\Qv_{\|}(\rho\Hv\Hv^H+\Iv_{N_R})\Qv_{\|}^H)&\leq&\text{Tr}(\Qv(\rho\Hv\Hv^H+\Iv_{N_R})\Qv^H).\nonumber
\eea
This result reveals that the optimal relay matrix only contains the parallel components to the column space of $\Hv$, i.e., $\hat{\Qv}=\Qv_{\|}$
that can be generally expressed as $\Qv_{\|}=\overline{\Bv}\Hv^H$ for any matrix $\overline{\Bv}\in\mathbb{C}^{N_R\times N_S}$
since the row space of $\Hv^H$ determines the row space of $\Qv_{\|}$.
Now, let us define $\Pv\in\mathbb{C}^{N_S\times N_S}$ as an arbitrary square invertible matrix.
Then, the optimal relay matrix can be further extended to $\hat{\Qv}=\Bv\Pv\Hv^H$ for any matrix
$\Bv\in\mathbb{C}^{N_R\times N_S}$, and thus the proof is completed.

\section{Achievability Proof of $d^*(r)$} \label{Appendix:Proof of Theorem1}
According to \cite{Tang:07}, the optimal relay matrix $\hat{\Qv}$ which maximizes the MI (\ref{eq:max MI}) is generally written by
$\Qv=\widetilde{\Uv}_g\Xv(\Iv_M+\rho\widetilde{\mathbf{\Lambda}}_h)^{-1/2}\widetilde{\Vv}_h^H$, where $\Xv\in\mathbb{C}^{M\times M}$ may be any matrix and
$\widetilde{\Vv}_h$ designates a matrix constructed by the first $M$ columns of the left singular matrix of $\Hv=\Vv_h\mathbf{\Lambda}_h^{1/2}\Uv_h^H$.
Then, it is shown that the MI maximization problem in (\ref{eq:max MI}) is equivalently changed to
\bea
\mathcal{I}&=&\max_{\mathbf{X}}\frac{1}{2}\log\frac{\left|\Iv_M+\rho\widetilde{\mathbf{\Lambda}}_h\right|\left|\Iv_M+\Xv^H\widetilde{\mathbf{\Lambda}}_g\Xv\right|}
{\left|\Iv_M+\rho\widetilde{\mathbf{\Lambda}}_h+\Xv^H\widetilde{\mathbf{\Lambda}}_g\Xv\right|}\nonumber\\
&&s.t.~~\text{Tr}(\Xv\Xv^H)\leq P_R=N_S\rho.\nonumber
\eea

Instead of the optimal solution \cite{Tang:07},
let us consider a suboptimal $\Xv=\sqrt{\rho}\Iv$ which satisfies the above power constraint.
Then, we obtain the MI lower-bound as follows:
\bea
\mathcal{I}\geq\frac{1}{2}\log\frac{\left|\Iv_M+\rho\widetilde{\mathbf{\Lambda}}_h\right|\left|\Iv_M+\rho\widetilde{\mathbf{\Lambda}}_g\right|}
{\left|\Iv_M+\rho\widetilde{\mathbf{\Lambda}}_h+\rho\widetilde{\mathbf{\Lambda}}_g\right|}~~~~~~~~~~~~~~~~~~~~~~~~~~~~~~\nonumber\\
=\frac{1}{2}\log\frac{\prod_{k=1}^M\left[(1+\rho^{1-\alpha_k})(1+\rho^{1-\beta_k})\right]}
{\prod_{k=1}^M\left(1+\rho^{1-\alpha_k}+\rho^{1-\beta_k}\right)}~~~~~~~~~~~~~~~~~~~~~~~~\nonumber\\
\doteq\frac{1}{2}\log\bigg(\prod_{k=1}^M\rho^{\max(0,1-\alpha_k,1-\beta_k,2-\alpha_k-\beta_k)-\max(0,1-\alpha_k,1-\beta_k)}\bigg)\nonumber\\
=\frac{1}{2}\log\bigg(\prod_{k=1}^M\rho^{\min[(1-\alpha_k)^+, (1-\beta_k)^+]}\bigg),~~~~~~~~~~~~~~~~~~~~~~~~~~~\nonumber
\eea
where $\alpha_k$ and $\beta_k$ are defined in (\ref{eq:alpha beta}), and
setting the target rate as $R(\rho)=r\log\rho$, it follows
\bea
P_{\text{out}}&\dot{\leq}&P\left(\frac{1}{2}\log\left(\prod_{k=1}^M\rho^{\min[(1-\alpha_k)^+, (1-\beta_k)^+]}\right)<R(\rho)\right)\nonumber\\
&=&P\left(\sum_{k=1}^N\min[(1-\alpha_k)^+, (1-\beta_k)^+]<2r\right)\nonumber\\
&=&P\left(\mathcal{E}\right)\nonumber
\eea
where $\mathcal{E}=\{\sum_{k=1}^N\min[(1-\alpha_k)^+, (1-\beta_k)^+]<2r\}$ denotes the outage event.
The second equality is due to the fact that we have $\rho\lambda_{g,k}=\rho^{1-\beta_k}=0$, i.e., $\beta_k>1$ for $k=N+1,\ldots,M$.
Now, defining $P(\mathcal{E})\doteq\rho^{-d(r)}$ and applying the Varadhan's lemma \cite{Kumar:09} \cite{Zheng:03} by using the pdfs of random vectors
$\av=[\alpha_1,\ldots,\alpha_M]$ and $\bv=[\beta_1,\ldots,\beta_{L=\min(N_R,N_D)}]$ given in (\ref{eq:pdf a}) and (\ref{eq:pdf b}),
we obtain the outage exponent as
\bea
d(r)&\doteq&\inf_{\{\mathbf{a},\mathbf{b}\}\in\mathcal{E}, \forall\alpha_k>0,\forall\beta_k>0}\sum_{i=1}^M(N_S+N_R-2i+1)\alpha_i\nonumber\\
&&~~~~~~~~~~~~~~~+~\sum_{i=1}^L(N_R+N_D-2i+1)\beta_i,\nonumber\\
&=&\min\bigg(\sum_{i=2r+1}^M(N_S+N_R-2i+1),\nonumber\\
&&~~~~~~~~~~~~~~~~~~~~~~\sum_{i=2r+1}^L(N_R+N_D-2i+1)\bigg)\nonumber\\
&=&(N_R-2r)(\min(N_S,N_D)-2r),\nonumber
\eea
and thus the proof is completed.

\section{Proof of Lemma \ref{Lemma:Lemma3}} \label{Appendix:Proof of Lemma3}

From the definition of $\Rv_y$ in (\ref{eq:relay power constraint}), it follows
\bea\label{eq:Ry2}
\Rv_y&=&(\Hv^H\Hv+\rho^{-1}\Iv_{N_S})^{-1}\Hv^H(\rho\Hv\Hv^H+\Iv_{N_S})\nonumber\\
&&\times\Hv(\Hv^H\Hv+\rho^{-1}\Iv_{N_S})^{-1}\nonumber\\
&\overset{(a)}{=}&\rho\Hv^H\Hv(\Hv^H\Hv+\rho^{-1}\Iv_{N_S})^{-1}\nonumber\\
&=&\rho(\Hv^H\Hv+\rho^{-1}\Iv_{N_S}-\rho^{-1}\Iv_{N_S})(\Hv^H\Hv+\rho^{-1}\Iv_{N_S})^{-1}\nonumber\\
&=&\rho\Iv_{N_S}-(\Hv^H\Hv+\rho^{-1}\Iv_{N_S})^{-1},
\eea
where we obtain (a) by invoking the matrix inversion lemma.
Since $\Av-\Bv=\Cv$ implies that $\Av\succeq\Cv$ for $\Av,\Bv,\Cv\in\mathbb{S}^+$,
it is obvious from (\ref{eq:Ry2}) that $\Rv_y\preceq\rho\Iv_{N_S}$ and the proof is completed.

\section{Proof of Lemma \ref{Lemma:Lemma4}} \label{Appendix:Proof of Lemma4}

By definition, $\Rv_z$ in (\ref{eq:Rz}) is rephrased as $\Rv_z=\rho\Iv_{N_S}+(\Hv^H\Hv)^{-1}$, from which it is immediate that $\rho\Iv_{N_S}\preceq\Rv_z$.
Meanwhile, since we have $\Hv^H\Hv\succeq\lambda_{h,N_S}\Iv_{N_S}$,
one can easily show that $\Rv_z\preceq(\rho+\lambda_{h,N_S}^{-1})\Iv_{N_S}=\rho(1+\rho^{\alpha_{N_S}-1})\Iv_{N_S}$.
From the Varadhan's lemma as in \cite{Zheng:03}, it is also true that $(1+\rho^{\alpha_{N_S}-1})\doteq\rho^0$, because
$\alpha_{N_S}$ is smaller than $1$ with probability $1$.
Therefore, we finally obtain $\rho\Iv_{N_S}\preceq\Rv_z\dot{\preceq}\rho\Iv_{N_S}$ and the proof is concluded.

\section{Eigenvalue Distribution of Rayleigh Product Channels}\label{Appendix:Rayleigh Product}
Let $\Gv$,$\Hv$ be $n\times l$, $l\times m$ independent matrices with i.i.d. entries distributed as $\mathcal{CN}(0,1)$.
We define a positive semi-definite matrix $\Av=\Hv^H\Gv^H\Gv\Hv$ with eigenvalues $\lambda_{t,1}>\cdots>\lambda_{t,N}$ and
$\gamma_k\triangleq-\log\lambda_{t,k}/\log\rho$ for $k=1,\ldots,N$. Then, denoting the pdf of a random vector
$\cv=[\gamma_1,\ldots,\gamma_N]$ as $f(\cv)\doteq\rho^{-\theta(\cv)}$, the exponential order $\theta(\cv)$ is given as follows:
\begin{itemize}
\item When $l<\min(m,n)$,
\bea\label{eq:fc1}
\theta_1(\cv)=\!\!\!\!\sum_{k=1}^{l-|m-n|}\left(\overline{n}+1-2k+\left\lfloor\frac{l+k+|m-n|}{2}\right\rfloor\right)\gamma_k\nonumber
\eea
\vspace{-10pt}
\bea
+\sum_{l-|m-n|+1}^l\left(\overline{n}+l+1-2k\right)\gamma_k
\eea
In this case, $m$ and $n$ can be exchanged by the reciprocity property of MIMO channels.
\item When $n\leq l\leq m$ or $m\leq l\leq n$,
\bea\label{eq:fc2}
\theta_2(\cv)=\!\!\!\!\sum_{k=1}^{l-|m-n|}\left(\overline{p}+1-2k+\left\lfloor\frac{l+k+|m-n|}{2}\right\rfloor\right)\gamma_k\nonumber
\eea
\vspace{-10pt}
\bea
+\sum_{l-|m-n|+1}^{\overline{p}}\left(\overline{n}+l+1-2k\right)\gamma_k
\eea
\item When $n\leq m< l$ or $m\leq n< l$,
\bea\label{eq:fc3}
\theta_3(\cv)=\!\!\!\!\sum_{k=1}^{\overline{q}-|l-\overline{p}|}\left(\overline{p}+1-2k+\left\lfloor\frac{\overline{q}+k+|l-\overline{p}|}{2}\right\rfloor\right)\gamma_k\nonumber
\eea
\vspace{-10pt}
\bea
+\sum_{\overline{q}-|l-\overline{p}|+1}^{\overline{p}}\left(m+n+1-2k\right)\gamma_k
\eea
\end{itemize}
where $\overline{p}\triangleq\min(m,n)$ and $\overline{q}\triangleq\max(m,n)$. For more details, please refer to \cite{SYang:11}.

\bibliographystyle{IEEEtr}

\input{bibliography.filelist}

\end{document}